\documentclass[final]{lmcs}
\pdfoutput=1
\usepackage[utf8]{inputenc}

\usepackage{lastpage}
\lmcsdoi{18}{3}{31}
\lmcsheading{}{\pageref{LastPage}}{}{}%
{Feb.~15,~2022}{Sep.~08,~2022}{}

\pdfoutput=1
\usepackage{ifthen}

\usepackage{amssymb}
\usepackage{xspace}
\usepackage{booktabs}
\usepackage{makecell}
\usepackage{microtype}
\usepackage{subcaption}
\usepackage{todonotes}
\usepackage{hyphenat}

\usepackage{twshowkeys}

\newboolean{withappendix}
\setboolean{withappendix}{true}

\newenvironment{notheorembrackets}{}{}

\newcommand{\hackynewpage}{%
}

\usepackage{tikz}
\usetikzlibrary{cd,fit,calc,backgrounds}

\tikzset{shiftarr/.style={
        rounded corners,%
        to path={--([#1]\tikztostart.center)
                     -- ([#1]\tikztotarget.center) \tikztonodes
                     -- (\tikztotarget)},
}}

\newsavebox{\mypullbackcorner}%
\sbox{\mypullbackcorner}{%
\begin{tikzpicture}
    \draw[-] (0,0) -- (.5em,.5em) -- (0,1em);
\end{tikzpicture}%
}
\newcommand{\pullbackangle}[2][]{\arrow[phantom,to path={
                     -- ($ (\tikztostart)!1cm!#2:([xshift=8cm]\tikztostart) $)
                        node[anchor=west,pos=0.0,rotate=#2,
                        inner xsep = 0]
                        {\begin{tikzpicture}[minimum
                        height=1mm,baseline=0,#1]
    \draw[-] (0,0) -- (.5em,.5em) -- (0,1em);
                        \end{tikzpicture}}}]{}}

\definecolor{coalgstate}{HTML}{2e3436}
\definecolor{coalgedge}{HTML}{4F677E}
\definecolor{coalgframe}{HTML}{B9CDE0}
\tikzset{
  coalgebra/.style={
    block line/.style={
      draw=black!50,
      line width=1.2pt,
    },
    block/.style={
      block line,
      rounded corners=3pt,
      inner sep=1pt,
      minimum height=6mm,
      minimum width=6mm,
    },
    scissors line/.style={
      draw=black!50,
      text=black!50,
      font=\footnotesize,
      line width=0.8pt,
      shorten <= -4pt,
      shorten >= -4pt,
      dotted,
    },
    state/.style={
      text depth=0pt,
      outer sep=0pt,
      inner sep=4pt,
      text=coalgstate,
    },
    add state borders/.style={
      state/.append style={
        draw=coalgstate,
        shape=circle,
        outer sep=3pt,
        line width=.4pt,
        inner sep=2.4pt, %
      },
    },
    final/.style={
      draw=coalgstate,
      double=white,
      inner sep=2pt,
      double distance=1pt,
      line width=.4pt,
    },
    transition/.style={
      -{latex},
      line width=0.8pt,
      draw=coalgedge,
      preaction = {draw,-,draw=white,line width=4.6pt,line cap=round},
      every node/.append style={
        text=coalgedge,
      },
    },
    path with edges/.style={
      every edge/.append style={transition}
    },
  },
  coalgebra frame/.style={
    draw=coalgframe,
    rounded corners=4pt,
    line width=1pt,
  },
  coalgname/.style={
    outer sep=5pt,
    inner sep=2pt,
    rounded corners=1pt,
    fill=coalgframe,
    text=black,
  },
}

\newcommand{\C}{\ensuremath{\mathcal{C}}}
\newcommand{\E}{\ensuremath{\mathcal{E}}}
\newcommand{\R}{\ensuremath{\mathbb{R}}}
\newcommand{\N}{\ensuremath{\mathbb{N}}}
\newcommand{\K}{\mathcal{K}}%
\newcommand{\D}{\ensuremath{\mathcal{D}}}
\newcommand{\CO}{\ensuremath{\mathcal{O}}}
\newcommand{\M}{\ensuremath{\mathcal{M}}}
\newcommand{\Set}{\ensuremath{\mathsf{Set}\xspace}}
\newcommand{\set}[2][]{\ensuremath{{#1\{#2#1\}}}}
\newcommand{\EM}{\ensuremath{\mathsf{EM}}}
\newcommand{\Pow}{\ensuremath{\mathcal{P}}}
\newcommand{\Bag}{\ensuremath{\mathcal{B}}}
\newcommand{\colim}{\ensuremath{\operatorname{colim}}}
\newcommand{\dual}{{\ensuremath{\operatorname{\mathsf{op}}}}}
\newcommand{\Coalg}{\ensuremath{\mathsf{Coalg}}}
\newcommand{\Alg}{\ensuremath{\mathsf{Alg}}}

\newcommand{\StrongEpi}{\ensuremath{\mathsf{StrongEpi}}}
\newcommand{\StrongMono}{\ensuremath{\mathsf{StrongMono}}}
\newcommand{\Mor}{\ensuremath{\mathsf{Mor}}}
\newcommand{\Iso}{\ensuremath{\mathsf{Iso}}}
\newcommand{\Epi}{\ensuremath{\mathsf{Epi}}}

\newcommand{\Mono}{\ensuremath{\mathsf{Mono}}}

\newcommand{\reach}{\ensuremath{\mathsf{reach}}}
\newcommand{\pr}{\ensuremath{\mathsf{pr}}}

\newcommand{\inj}{\ensuremath{\mathsf{inj}}}
\newcommand{\id}{\ensuremath{\mathsf{id}}}
\newcommand{\Id}{\ensuremath{\mathsf{Id}}}
\renewcommand{\Im}{\ensuremath{\mathsf{Im}}}

\newcommand{\monoto}{\ensuremath{\rightarrowtail}}
\newcommand{\hookto}{\ensuremath{\hookrightarrow}}
\newcommand{\epito}{\ensuremath{\twoheadrightarrow}}

\newcommand{\takeout}[1]{\relax}

\newcommand{\itemref}[2]{\autoref{#1}.\ref{#2}}
\newcommand{\textqt}[1]{`#1'}
\newcommand{\singlequote}[1]{`#1'}
\newcommand{\etal}{\text{et al.}}
\newcommand{\Adamek}{Ad{\'{a}}mek}

\newcommand{\textshaded}[1]{\text{#1}}

\newenvironment{proofappendix}[2][Proof of]{%
\subsubsection*{#1~\autoref{#2}}%
\addcontentsline{toc}{subsection}{#1~\autoref{#2}}%
}{}

\newenvironment{listinenv}{\setlength{\parskip}{0.2em plus 0.1em minus
    0em}\leavevmode}%
{\setlength{\parskip}{0.0em plus 0.6em minus 0.0em}}

\makeatletter 
\newcommand\mynobreakpar{\par\nobreak\@afterheading} 
\makeatother

\newcommand{\descto}[3][]{\arrow[phantom]{#2}[#1]{\text{\footnotesize{}\begin{tabular}{c}#3\end{tabular}}}}

\theoremstyle{definition}

\newtheorem{assumption}[thm]{Assumption}
\newtheorem{instance}[thm]{Instance}
\newtheorem{definition}[thm]{Definition}

\newtheorem{example}[thm]{Example}
\newtheorem{remark}[thm]{Remark}
\newtheorem*{rem*}{Remark}

\theoremstyle{plain}

\newtheorem{proposition}[thm]{Proposition}
\newtheorem{lemma}[thm]{Lemma}
\newtheorem{corollary}[thm]{Corollary}

\title[Minimality Notions
via Factorization Systems
and Examples]{Minimality Notions
  via Factorization Systems\texorpdfstring{\\}{}
  and Examples}

\author{Thorsten Wißmann\lmcsorcid{0000-0001-8993-6486}}
\address{Radboud University, Nijmegen, The Netherlands}
\urladdr{https://thorsten-wissmann.de} %


\thanks{Supported by the NWO TOP project 612.001.852}%

\begin{document}
\maketitle
\begin{abstract}
  For the minimization of state-based systems (i.e.\ the reduction of the number
  of states while retaining the system's semantics), there are two obvious
  aspects: removing unnecessary states of the system and merging redundant
  states in the system. In the present article, we relate the two minimization aspects on
  coalgebras by defining an abstract notion of minimality.

  The abstract notions minimality and minimization live in a general category
  with a factorization system. We will find criteria on the category that ensure
  uniqueness, existence, and functoriality of the minimization aspects. The
  proofs of these results instantiate to those for reachability and observability minimization
  in the standard coalgebra literature. Finally, we will see how the two aspects
  of minimization interact and under which criteria they can be sequenced in any
  order, like in automata minimization.

  This is an updated version that fixes a mistake in \autoref{figTSWell},
  spotted by Bálint Kocsis.
\end{abstract}

\section{Introduction}
Minimization is a standard task in computer science that comes in different
aspects and lead to various algorithmic challenges. The task is to reduce the
size of a given system while retaining its semantics, and in general there are
two aspects of making the system smaller: 1.~merge redundant parts of the system
that exhibit the same behaviour (\emph{observability}) and 2.~omit
unnecessary parts (\emph{reachability}). Hopcroft's automata minimization
algorithm~\cite{Hopcroft71} is an early example: in a given deterministic
automaton, 1.~states accepting the same language are identified and
2.~unreachable states are removed. Moreover, Hopcroft's algorithm runs in
quasilinear time; for an automaton with $n$ states, reachability is computed 
in $\CO(n)$ and observability in $\CO(n\log n)$.

Since the reachability is a simple depth-first search, it is straightforward to
apply it to other system types. On the other hand, it took decades until
quasilinear minimization algorithms for observability were developed for other
system types such as transition systems~\cite{PaigeT87}, labelled transition
systems~\cite{DovierPP04,Valmari09}, or Markov
chains~\cite{DerisaviHS03,ValmariF10}. Despite their differences in complexity,
the aspects of observability and reachability have very much in common when
modelling state-based systems as coalgebras. Then, observability is the task to
find the greatest coalgebra quotient and reachability is the task of finding the
smallest subcoalgebra containing the initial state, or generally, a
distinguished point of interest.

In the present article, we define an abstract notion of minimality and
minimization in a category with an $(\E,\M)$-factorization system. Such a
factorization systems gives rise to a generalized notion of quotients and subobjects.
Then, \textqt{minimization} is the task of finding the least quotient resp.~subobject.
To make this general setting applicable to coalgebras, we show that the category of
coalgebras inherits the factorization system from the base category under a mild
assumption -- namely that the functor preserves $\M$.
Dually, if the functor preserves $\E$, a factorization system also lifts to algebras, and even to the
Eilenberg-Moore category. 

Then, we will present different characterizations of minimality
(\autoref{figCharacterizations}) and then study properties of minimizations,
e.g.~under which criteria they exist and are unique, rediscovering the
respective proofs for reachability and observability for coalgebras in the
literature~\cite{amms13,Gumm03}. When combining the two minimization aspects, we
discuss under which criteria reachability and observability can be computed in
arbitrary order.

The goal of the present work is not only to show the connections between
existing minimality notions, but also to provide a series of basic results that
can be used when developing new minimization techniques or even new notions of
minimality.

\paragraph*{Related work}
There is a series of
works~\cite{BidoitHK01,BezhanishviliKP12,BonchiBHPRS14,Rot16} that studies the
minimization of coalgebras by their duality to algebras. In those works, the
correspondence between observability in coalgebras and reachability in algebras
is used. For instance, Rot~\cite{Rot16} relates the final sequence (for
observability in coalgebras) with the initial sequence (for reachability in
algebras). In the present paper however, we consider both observability and
reachability on an abstract level that work for a general factorization system
and discuss their instance in coalgebras.

The present article is an extended version of a conference paper~\cite{Wissmann21}, which
itself was based on Chapter 7 of the author's PhD
dissertation~\cite{Wissmann2020}. In the present version, the overall
presentation is extended with illustrated examples. Also, the proof of
\autoref{monadfactor} is simpler now.

\paragraph*{Structure of the paper}
First, preliminary definitions for (co)algebras and factorization systems are
recalled (\autoref{sec:prelim}). Then, these two notions are
brought together by showing that the factorization system lifts to coalgebras
and Eilenberg\hyp{}Moore algebras under mild assumptions
(\autoref{secCoalgFact}). Thus, we can define categorical notions of minimality and
minimization using only factorization systems, which then also apply to
coalgebras (\autoref{secMinObj}), yielding minimality notions of reachability
and observability. We finally investigate their interplay, if coalgebras are
minimized under both minimality notions (\autoref{secInterplay}).

All results that are part of the present work are proven in the main text. For a
couple of well-known standard results, we recall the proofs in the appendix for
the convenience of the reader.

\section{Preliminaries}
\label{sec:prelim}
In the following, we assume basic knowledge of category theory~(cf.~standard
textbooks~\cite{joyofcats,awodey2010category}).

Given a diagram $D\colon \D\to \C$ (i.e.~a functor $D$ from a small category
$\D$), we denote its limit by $\lim D$ and colimit by $\colim D$ -- if they
exist. The limit projections, resp.~colimit injections, are denoted by
\[
  \pr_i\colon \lim D\to Di
  \qquad
  \inj_{i}\colon Di\to \colim D
  \qquad
  \text{for }i\in \D.
\]

\subsection{Coalgebra}

We model state-based systems as
coalgebras for an endofunctor $F\colon \C\to \C$ on a category $\C$:

\medskip
\noindent\begin{minipage}[c]{.8\textwidth}
\begin{definition}
  An \emph{$F$-coalgebra} (for an endofunctor $F\colon \C\to \C$) is a pair
  $(C,c)$ consisting of an object $C$ (of $\C$) and a morphism $c\colon C\to FC$
  (in $\C$). An $F$-coalgebra morphism $h\colon (C,c)\to (D,d)$ between
  $F$-coalgebras $(C,c)$ and $(D,d)$ is a morphism $h\colon C\to D$ with $d\cdot
  h = Fh\cdot c$.
\end{definition}
\end{minipage}
\hfill
\begin{tikzcd}[baseline=-2mm]
  |[alias=C]|
  C \arrow{d}{h}\arrow{r}{c}& FC \arrow{d}[overlay]{Fh}\\
  D \arrow{r}{d} & FD
\end{tikzcd}

\medskip
Intuitively, the \emph{carrier} $C$ of a coalgebra $(C,c)$ is the state space
and the morphism $c\colon C\to FC$ sends states to
their possible next states. The functor of choice $F$ defines how these possible
next states $FC$ are structured. Before discussing the role of the $F$-coalgebra
morphisms, let us list what $F$-coalgebras are for standard examples of functors $F$:

\newcommand{\drawDfaC}{%
      \node[state,final] (q) at (0,1) {$q$};
      \node[state,final] (p) at (1,1) {$p$};
      \node[state] (r) at (1,0) {$r$};
      \node[state,final] (s) at (0,0) {$s$};
      \begin{scope}[on background layer]
        \draw[transition,bend left] (q) to node[above] {$a$} (p);
        \draw[transition,bend right] (q) to node[outer sep=0pt,pos=0.3,above right] {$b$} (r);
        \draw[transition,bend right] (p) to node[outer sep=0pt,pos=0.3,left] {$b$} (r);
        \draw[transition,loop right,out=30,in=70,looseness=4]
        (p) to node[right] {$a$}(p);
        \draw[transition,bend left] (s) to node[outer sep=0pt,pos=0.3,left] {$a$} (q);
        \draw[transition,bend right] (s) to node[outer sep=0pt,pos=0.3,below] {$b$} (r);
        \draw[transition,bend right] (r) to node[outer sep=0pt,pos=0.5,right] {$a$} (p);
        \draw[transition,loop right,out=-70,in=-30,looseness=4]
        (r) to node[right] {$b$}(r);
      \end{scope}
}

\newcommand{\drawMarkovExLoop}{%
  \node[state] (q) at (0,1) {$q$};
  \node[state] (p) at (1,1) {$p$};
  \node[state] (r) at (1,0) {$r$};
  \node[state] (s) at (0,0) {$s$};
  \draw[transition,bend left] (q) to node[above] {-2} (p);
  \draw[transition,bend left] (p) to node[below,pos=0.3] {2} (q);
  \draw[transition,bend left] (q) to node[right,pos=0.7] {3} (s);
  \draw[transition,bend left] (s) to node[left] {5} (q);
  \draw[transition,bend left] (p) to node[right] {3} (r);
  \draw[transition,bend left] (r) to node[above] {1} (s);
}

\begin{figure}[b]
  \tikzset{
    figure coalgebra frame/.style={
      coalgebra frame,
      minimum width=3.4cm,
      minimum height=3cm,
    },
  }
  \begin{subfigure}{.33\textwidth} \centering
    \begin{tikzpicture}[coalgfit/.style={
        figure coalgebra frame,
        inner xsep=5mm,
        inner ysep=4mm,
        xshift=0mm,
        yshift=0mm,
      },
      x=13mm,y=13mm,
      ]
      \begin{scope}[coalgebra,shift={(0,0)},add state borders]
        \drawDfaC
        \node[coalgfit,fit={(s) (q) (p)}] (C) {};
      \end{scope}
    \end{tikzpicture}
    \hspace*{5mm} %
  \end{subfigure}
  \begin{subfigure}{.30\textwidth} \centering
    \begin{tikzpicture}[coalgebra, add state borders,x=14mm,y=14mm]
      \node[state] (q0) at (0,1) {$q$};
      \node[state] (q1) at (1.0,1.0) {$p$};
      \node[state] (q2) at (1.0,0.0) {$r$};
      \node[state] (s) at (0,0.0) {$s$};
      \draw[transition] (q0) to (q1);
      \draw[transition] (q0) to (s);
      \draw[transition] (q0) to (q2);
      \draw[transition] (q1) to (q2);
      \draw[transition] (s) to (q2);
      \draw[transition,loop right] (q1) to node[alias=loopnode]{} (q1);
      \node[figure coalgebra frame,fit={([xshift=-3mm]q0.west) (q1) (q2) (loopnode)}] {};
    \end{tikzpicture}
    \hspace*{5mm} %
  \end{subfigure}
  \begin{subfigure}{.30\textwidth} \centering
    \begin{tikzpicture}[coalgebra, add state borders, x=14mm, y=14mm]
      \drawMarkovExLoop
      \node[figure coalgebra frame,inner xsep=5mm,inner ysep=4mm,yshift=2mm,fit={(q) (p) (r) (s)}] {};
    \end{tikzpicture}
    \hspace*{3mm} %
  \end{subfigure}
  \\
  \begin{subfigure}{.30\textwidth} \centering
    \caption{$FX=2\times X^{\{a,b\}}$}
    \label{figExDfa}
  \end{subfigure}
  \begin{hideshowkeys}
  \begin{subfigure}{.30\textwidth} \centering
    \caption{$FX=\Pow X$}
    \label{figExTs}
  \end{subfigure}
  \begin{subfigure}{.30\textwidth} \centering
    \caption{$FX=(\R,+,0)^{(X)}$}
    \label{figExMarkov}
  \end{subfigure}
  \end{hideshowkeys}

  \caption{Examples of $F$-coalgebras for different $\Set$-functors $F$}
\end{figure}
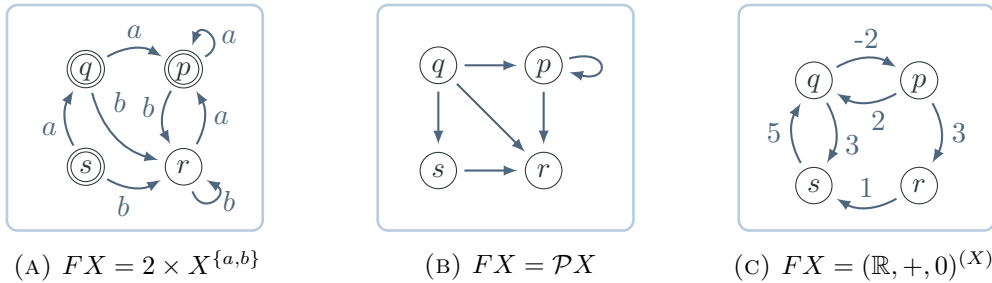

\begin{example}
  Many well-known system-types can be phrased as coalgebras:
  \begin{enumerate}
  \item Deterministic automata (without an explicit initial state) are
    coalgebras for the \Set-functor $FX = 2\times X^A$, where $A$ is the set of
    input symbols. In an $F$-coalgebra $(C,c)$, the first component of $c(x)$
    denotes the finality of the state $x\in C$ and the second component is the
    transition function $A\to C$ of the automaton.
    An example DFA for $A=\{a,b\}$ is shown in \autoref{figExDfa}, where the
    coalgebra map $c\colon C\to 2\times C^{\{a,b\}}$ is defined by:
    \[
      \begin{array}{l@{\qquad}l}
      c(q) = (1,(a\mapsto p,~b\mapsto r))
        &
      c(p) = (1,(a\mapsto p,~b\mapsto r))
        \\
      c(s) = (1,(a\mapsto q,~b\mapsto r))
        &
      c(r) = (0,(a\mapsto p,~b\mapsto r))
      \end{array}
    \]
    In general, the carrier of a coalgebra is not required to be finite; the
    carrier $C$ may be an arbitrary set.

  \item Labelled transition systems are coalgebras for the \Set-functor
    $FX=\Pow(A\times X)$.
    Coalgebras for the powerset functor $FX=\Pow X$ are transition systems
    (i.e.~for a singleton label set). An example of a $\Pow$-coalgebra is
    illustrated in \autoref{figExTs}; here, the coalgebra structure $c\colon
    C\to \Pow C$ is defined by:
    \[
      c(q) = \set{p, r, s}
      \qquad
      c(p) = \set{p, r}
      \qquad
      c(s) = \set{r}
      \qquad
      c(r) = \emptyset
    \]
    The successor structures are not
    ordered, so $\set{p,r}$ and $\set{r,p}$ are the same successor structure.
  \item Weighted systems with weights in a commutative monoid $(M,+,0)$ (and
    finite branching) are coalgebras for the \emph{monoid-valued functor}~\cite[Def.~5.1]{GummS01}
    $M^{(-)}\colon \Set\to\Set$ by
    \[
      M^{(X)} = \{\mu \colon X\to M\mid \mu(x) =0\text{ for all but finitely many
      }x\in X\}
    \]
    which sends a map $f\colon X\to Y$ to the map
    \[
    M^{(f)}\colon M^{(X)}\to M^{(Y)}
    \qquad M^{(f)}(\mu)(y) =\sum \{ \mu(x) \mid x\in X, f(x) = y\}.
    \]
    In an $M^{(-)}$-coalgebra $(C,c)$, the transition weight from state $x\in C$
    to $y\in C$ is given by $c(x)(y) \in M$, and a weight of 0 means that there
    is no transition.
    E.g.~one obtains real-valued weighted systems as coalgebras for the functor
    $(\R,+,0)^{(-)}$. \autoref{figExMarkov} illustrates a coalgebra $c\colon
    C\to (\R,+,0)^{(C)}$ that is defined by:
    \[
      \def\myzero{\color{black!50}0}
      \begin{array}{r@{q \mapsto~}r@{,\hspace*{3mm}p \mapsto~}r@{,\hspace*{3mm}r \mapsto~}r@{,\hspace*{3mm}s \mapsto~}r@{}l}
        c(q) = (&\myzero & -2 & \myzero & 3 & ) \\
        c(p) = (&2 & \myzero & 3 &\myzero &) \\
        c(r) = (&\myzero & \myzero & \myzero &1 &) \\
        c(s) = (&5 & \myzero & \myzero &\myzero &) \\
      \end{array}
    \]
  \item The \emph{bag} functor is defined
    by $\Bag X = (\N,+,0)^{(X)}$. Equivalently, $\Bag X$ is the set of finite
    multisets on $X$.
    Its coalgebras can be viewed as weighted systems (i.e.~via the submonoid
    inclusion $\N \subseteq \R$) or as transition systems
    in which there can be more than one transition between two states.
  \item A wide range of probabilistic and weighted systems can be obtained as
    coalgebras for respective distribution functors, see e.g.~Bartels
    \etal~\cite{BartelsSV04}.
  \end{enumerate}
\end{example}

\begin{definition}
  The category of $F$-coalgebras and their morphisms is denoted by $\Coalg(F)$.
\end{definition}

Intuitively, the coalgebra morphisms preserve the behaviour of states:
\begin{definition}
  \label{defBehEq}
  In \Set, two states $x,y\in C$ in an $F$-coalgebra $(C,c)$ are
  \emph{behaviourally equivalent} if there is a coalgebra homomorphism
  $h\colon (C,c)\to (D,d)$ with $h(x) = h(y)$.
\end{definition}

\newcommand{\drawDfaImh}{%
      \node[state,final] (p) at (1,1) {$\bar p$};
      \node[state] (r) at (1,0) {$r$};
      \node[state,final] (s) at (0,0) {$s$};
      \begin{scope}[on background layer]
        \draw[transition,bend right] (p) to node[outer sep=0pt,pos=0.3,left] {$b$} (r);
        \draw[transition,loop right,out=30,in=70,looseness=4]
        (p) to node[right] {$a$}(p);
        \draw[transition,bend left] (s) to node[outer sep=0pt,pos=0.3,left] {$a$} (p);
        \draw[transition,bend right] (s) to node[outer sep=0pt,pos=0.3,below] {$b$} (r);
        \draw[transition,bend right] (r) to node[outer sep=0pt,pos=0.5,right] {$a$} (p);
        \draw[transition,loop right,out=-70,in=-30,looseness=4]
        (r) to node[right] {$b$}(r);
      \end{scope}
  }
\newcommand{\drawDfaD}{%
  \node[state] (t) at (0,1) {$t$};
  \drawDfaImh
  \begin{scope}[on background layer]
    \draw[transition,bend right] (t) to node[outer sep=0pt,pos=0.3,left] {$a$} (s);
    \draw[transition,bend left] (t) to node[outer sep=0pt,pos=0.5,above] {$b$} (p);
  \end{scope}
  }

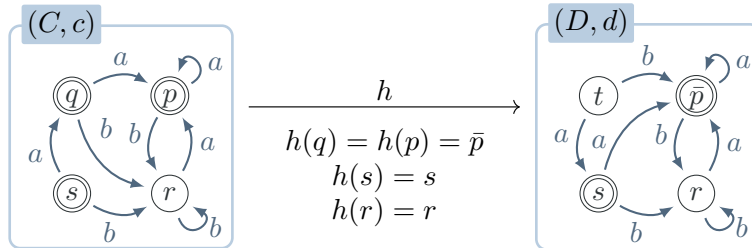
\begin{figure}[b]
    \begin{tikzpicture}[coalgfit/.style={
        coalgebra frame,
        inner xsep=5mm,
        inner ysep=5mm,
        xshift=0mm,
        yshift=1mm,
      },
      x=13mm,y=13mm,
      ]
    \begin{scope}[coalgebra,shift={(0,0)},add state borders]
      \drawDfaC
      \node[coalgfit,fit={(s) (q) (p)}] (C) {};
      \node[coalgname,anchor=west] at (C.north west) {$(C,c)$};
    \end{scope}
    \begin{scope}[coalgebra,shift={(7cm,0)},add state borders]
      \drawDfaD
      \node[coalgfit,fit={(s) (r) (p)}] (D) {};
      \node[coalgname,anchor=west] at (D.north west) {$(D,d)$};
    \end{scope}
    \begin{scope}[path/.style={
        commutative diagrams/.cd, every arrow, every label,
        shorten <= 2mm,shorten >= 2mm,
      }]
      \path[path,yshift=4mm] (C) edge[->]
      node[font=\normalsize] {$h$}
      node[font=\normalsize,align=center,below,yshift=-2mm] {$h(q) = h(p) = \bar p$
        \\ $h(s) = s$
        \\ $h(r) = r$
      }
      (D) ;
    \end{scope}
  \end{tikzpicture}
  \caption{Examples of an $F$-coalgebra morphisms for $FX=2\times X^{\{a,b\}}$}
  \label{figExDfaHom}
\end{figure}

\begin{example}
  For the running examples of $F$, coalgebraic behavioural equivalence
  instantiates to well-known system equivalences:
  \begin{enumerate}
  \item For deterministic automata ($FX=2\times X^A$), the coalgebra morphism square
    means that a coalgebra morphism $h\colon (C,c)\to (D,d)$ has
    to preserve the finality of states and the transition function:
    \[
      q\text{ final ~iff~ }h(q)\text{ final}
      \quad
      \text{and}
      \quad
      \begin{tikzpicture}[coalgebra,baseline=(q.base)]
        \node (q) {\(q\)};
        \node (p) at (1.4,0) {\(q'\)};
        \begin{scope}[on background layer]
          \draw[transition] (q) to node[above] {$a$} (p);
        \end{scope}
      \end{tikzpicture}
      \text{ iff }
      \begin{tikzpicture}[coalgebra,baseline=(q.base)]
        \node (q) {\(h(q)\)};
        \node (p) at (1.6,0) {\(h(q')\)};
        \begin{scope}[on background layer]
          \draw[transition] (q) to node[above] {$a$} (p);
        \end{scope}
      \end{tikzpicture}
      \quad
      \text{for all }q,q'\in C, a\in A.
    \]
    An example of a coalgebra morphism is illustrated in \autoref{figExDfaHom}.
    We use bars in state names to indicate that two states of a
    coalgebra were merged into one state in the codomain.
    Here, the states $q$ and $p$ are identified, showing that they are
    behaviourally equivalent. 
    However, a coalgebra homomorphism does not need to
    identify all states of equivalent behaviour, and indeed $h$ does not
    identify $s$ with $q$ and $p$ even though $s$ has the same behaviour. Also,
    the codomain may have additional states, e.g.~$t\in D$ is not in the
    image of $h$.

    In general, states in a coalgebra for $FX=2\times X^A$ are behaviourally
    equivalent iff they accept the same language~\cite[Example 9.5]{Rutten00}.
    For example, $s,p,q\in C$ in \autoref{figExDfaHom} accept all words in
    $\{a,b\}^*$ that do not end in $b$.

    The argument for the correspondence between behavioural equivalence and
    language equivalence is roughly as follows.
    For sufficiency, if two states $x,y$ in an $F$-coalgebra are
    identified by a coalgebra homomorphism, then one can show by induction over
    input words $w\in A^*$ that either both states or neither of them accepts
    $w$. For necessity, consider the map
    \[
      g\colon C\to \Pow(A^*)
      \qquad
      g(q) = \{w\in A^*\mid q\xrightarrow{w}q'\text{ and $q'$ final}\}
    \]
    which sends states to their semantics.
    This map $g$ is a coalgebra homomorphism for the $F$-coalgebra structure
    $p\colon \Pow(A^*)\to 2\times \Pow(A^*)^A$ given by
    \[
      \pr_1(p(L)) =\begin{cases}
        1 &\text{if }\varepsilon\in L \\
        0 &\text{if }\varepsilon\notin L
      \end{cases}
      \qquad
      \pr_2(p(L))(a) = \{w\mid a\,w\in L\}.
    \]
    In fact, $(\Pow(A^*), p)$ is the \emph{final} $F$-coalgebra. Final
    coalgebras give rise to a coinduction principle and are related to the
    minimality of coalgebras, but they are not needed for the present article.
    Thus we refer to the standard coalgebraic literature~\cite{Rutten00,JacobsR97,Adamek05,Jacobs17} for
    further details on final coalgebras and coinduction.
  \item For labelled transition systems ($FX=\Pow(A\times X)$), states are
    behaviourally equivalent iff they are bisimilar~\cite{aczelmendler:89}.
  \item 
    For weighted systems, i.e.~coalgebras for $M^{(-)}$, the coalgebraic
    behavioural equivalence captures weighted bisimilarity~\cite{KlinS13}.

    Explicitly, a map $h\colon C\to D$ is a coalgebra morphism $h\colon
    (C,c)\to (D,d)$, iff
    \[
      d(h(q))(p)
      = \sum\set[\big]{c(q)(q')\mid q'\in C, h(q') = p}
      \qquad\text{for all }q\in C, p\in D.
    \]
    So, whenever two successor states are merged, then their transition weights
    are summed up (\autoref{figExMarkovHomSimple}).
    \begin{figure}[t]
      \begin{tikzpicture}
        \begin{scope}[coalgebra, add state borders, x=14mm, y=14mm]
          \node[state] (q) {$q$};
          \node[state] (p) at (1.8,0.3) {$p$};
          \node[state] (r) at (1.8,-0.3) {$r$};
          \draw[transition] (q) to node[sloped,above] {$m_1$} (p);
          \draw[transition] (q) to node[sloped,below] {$m_2$} (r);
          \node[coalgebra frame,inner xsep=5mm,inner ysep=0mm,
          minimum height=8mm,
          yshift=0mm,fit={(q) (p) (r)}] (C) {};
          \node[coalgname,anchor=west] at (C.north west) {$(C,c)$};
        \end{scope}
        \begin{scope}[coalgebra, shift={(80mm,0mm)}, add state borders, x=14mm, y=14mm]
          \node[state] (q) {$q$};
          \node[state] (r) at (1.8,0) {$\bar p$};
          \draw[transition] (q) to node[sloped,above] {$m_1+m_2$} (r);
          \node[coalgebra frame,inner xsep=5mm,inner ysep=0mm,
          minimum height=15mm,
          fit={(q) (r)}] (D) {};
          \node[coalgname,anchor=west] at (D.north west) {$(D,d)$};
        \end{scope}
        \begin{scope}[path/.style={
            commutative diagrams/.cd, every arrow, every label,
            shorten <= 2mm,shorten >= 2mm,
          }]
          \path[path,yshift=4mm] (C) edge[->]
          node[font=\normalsize] {$h$}
          node[font=\normalsize,align=center,below,yshift=-2mm]
          {$h(q) = q$
            \\ $h(p) = h(r) = \bar p$
          }
          (C -| D.west) ;
        \end{scope}
      \end{tikzpicture}
      \caption{Simple example of an $M^{(-)}$-coalgebra mormphism}
      \label{figExMarkovHomSimple}
    \end{figure}
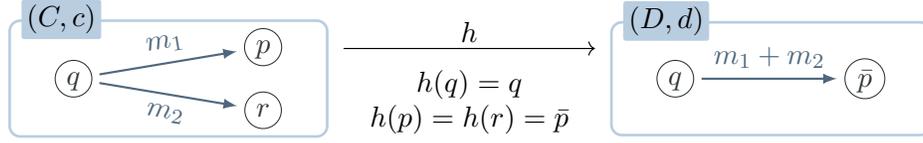
    \begin{figure}[b]
      \begin{tikzpicture}
        \begin{scope}[coalgebra, add state borders, x=14mm, y=14mm, bend angle=15]
          \drawMarkovExLoop
          \node[coalgebra frame,inner xsep=5mm,inner ysep=2mm,
          yshift=1mm,fit={(q) (p) (r) (s)}] (C) {};
          \node[coalgname,anchor=west] at (C.north west) {$(C,c)$};
        \end{scope}
        \begin{scope}[coalgebra, shift={(80mm,7mm)}, add state borders, x=14mm, y=14mm]
          \node[state] (q) {$\bar q$};
          \node[state] (s) at (1,0) {$\bar s$};
          \draw[transition,bend left] (q) to node[above] {1} (s);
          \draw[transition,bend left] (s) to node[below] {5} (q);
          \node[coalgebra frame,inner xsep=5mm,inner ysep=5mm,
          fit={(q) (s)}] (D) {};
          \node[coalgname,anchor=west] at (D.north west) {$(D,d)$};
        \end{scope}
        \begin{scope}[path/.style={
            commutative diagrams/.cd, every arrow, every label,
            shorten <= 2mm,shorten >= 2mm,
          }]
          \path[path,yshift=4mm] (C) edge[->]
          node[font=\normalsize] {$h$}
          node[font=\normalsize,align=center,below,yshift=-2mm]
          {$h(q) = h(r) = \bar q$
            \\ $h(p) = h(s) = \bar s$
          }
          (C -| D.west) ;
        \end{scope}
      \end{tikzpicture}
      \caption{Example of a $(\R,+,0)^{(-)}$-coalgebra morphism}
      \label{figExMarkovMor}
    \end{figure}
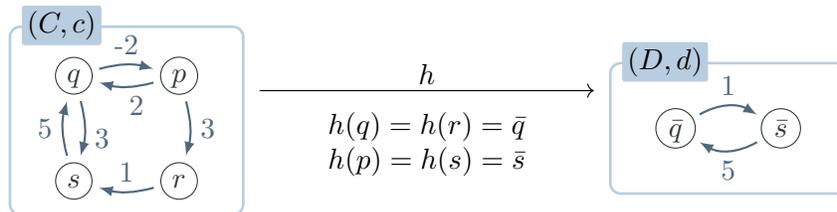

    For $M=(\R,+,0)$, an example of a $\R^{(-)}$-coalgebra morphism $h\colon
    (C,c)\to (D,d)$ is given in \autoref{figExMarkovMor}.  The map $h$ is a
    coalgebra morphism, because all transition weights in $(D,d)$
    are the sum of the corresponding transition weights~in~$(C,c)$:
    \[
      \begin{array}{c}
      \overbrace{d(\bar q)(\bar s)}^{1}
      = \overbrace{c(q)(p)}^{-2} + \overbrace{c(q)(s)}^{3} 
        \qquad
      \overbrace{d(\bar q)(\bar s)}^{1}
      = \overbrace{c(r)(p)}^{0} + \overbrace{c(r)(s)}^{1} 
      \\
        \underbrace{d(\bar s)(\bar q)}_{5}
        = \underbrace{c(s)(q)}_{5} + \underbrace{c(s)(r)}_{0} 
        \qquad
        \underbrace{d(\bar s)(\bar q)}_{5}
        = \underbrace{c(p)(q)}_{2} + \underbrace{c(p)(r)}_{3} 
      \end{array}
    \]
  \item Further semantic notions can be modelled with coalgebras by changing the
    base category from $\C=\Set$ to the Eilenberg-Moore~\cite{TuriP97} or Kleisli
    category~\cite{HasuoJS06} of a monad, to nominal
    sets~\cite{KurzPSV13,MiliusSW16}, or to partially ordered sets~\cite{BalanK11}.
  \end{enumerate}
\end{example}

The category of coalgebras inherits many properties from the base-category
$\C$. The categories are related via the forgetful functor
\[
  U\colon \Coalg(F)\to \C
  \qquad
  U(C,c) = C
  \qquad
  Uh = h
\]
which sends coalgebras to their carrier and coalgebra morphisms to their underlying morphism.
For instance, it is a standard result that $U$ creates colimits. Its proof is recalled in the
appendix for the convenience of the reader.
\begin{lemma}
  \label{cor:coalg-colims}
  The forgetful functor $U\colon \Coalg(F)\to \C$ creates all colimits. That is, the
  colimit of a diagram $D\colon \D\to \Coalg(F)$ exists, if $U\cdot D\colon
  \D\to \C$ has a colimit, and moreover, there is a unique coalgebra structure on
  $\colim(U\cdot D)$ making it the colimit of $D$ and making the colimit injections of
  $\colim D$ coalgebra morphisms.
\end{lemma}

On the other hand, we do not necessarily have all limits in $\Coalg(F)$. If $F$
preserves a limit of a diagram $U\cdot D$ for $D\colon \D\to \Coalg(F)$, then
the limit also exists in $\Coalg(F)$.

Coalgebras model systems with a transition structure, and pointed coalgebras
extend this by a notion of initial state:
\begin{definition}
  For an object $I\in \C$, an $I$-\emph{pointed} $F$-coalgebra $(C,c,i_C)$ is an
  $F$-coalgebra $(C,c)$ together with a morphism $i_C\colon I\to C$. A pointed
  coalgebra morphism $h\colon (C,c,i_C)\to (D,d,i_D)$ is a coalgebra morphism
  $h\colon (C,c)\to (D,d)$ that preserves the point: $i_D= h \cdot i_C$.

  The category of $I$-pointed $F$-coalgebras is denoted by $\Coalg_I(F)$.
\end{definition}

\begin{example} \label{exPointedDFA}
  For $I:= 1$ in $\Set$, a pointed coalgebra $(C,c,i_C)$ for $FX=2\times X^A$ is
  a deterministic automaton, where the initial state is given by the map
  $i_C\colon 1\to C$.
\end{example}

The point can also be understood as a (nullary) algebraic operation. In general, coalgebras
are dual to $F$-algebras in the following sense.

\medskip
\noindent\begin{minipage}[c]{.8\textwidth}
\begin{definition}
  An $F$-algebra (for a functor $F\colon \C\to\C$) is a morphism $a\colon FA\to
  A$, an algebra homomorphism $h\colon (A,a)\to (B,b)$ is a morphism $h\colon
  A\to B$ fulfilling $b\cdot Fh = h\cdot a$. The category of $F$-algebras is
  denoted by $\Alg(F)$.
\end{definition}
\end{minipage}
\hfill
\begin{tikzcd}[baseline=-2mm]
  FA
  \arrow{r}{a}
  \arrow{d}{Fh}
  & A
  \arrow{d}{h}
  \\
  FB \arrow{r}{b} & B
\end{tikzcd}
\medskip

The theory of algebras for a functor is dual to coalgebras in the sense that
$\Alg(F) = \Coalg(F^\dual)^\dual$ for $F^\dual\colon
\C^\dual\to\C^\dual$. The $I$-pointed coalgebras thus are also algebras for the
constant $I$ functor.
Most of the results of the present paper also apply to algebras for a functor
$F\colon \C\to\C$.

\subsection{Factorization Systems}

The process of minimizing a system constructs a quotient or subobject of the
state space, where the notions of quotient and subobject respectively stem from
a factorization system in the category of interest. This generalizes the well-known
image factorization of a function into a surjective and an injective part:

\begin{notheorembrackets}
\begin{defiC}[{\cite[Definition 14.1]{joyofcats}}] \label{D:factSystem}
  Given classes of morphisms $\E$ and $\M$ in $\C$, we say that $\C$ has an
  \emph{$(\E,\M)$-factorization system}\index{EM@$(\E,\M)$}\index{factorization system} provided that:
  \begin{enumerate}
  \item $\E$ and $\M$ are closed under composition with isomorphisms.
  \item
    \begin{minipage}[t]{.7\textwidth}
    Every morphism $f\colon A\to B$ in $\C$ has a factorization $f =
    m\cdot e$ with $e\in \E$ and $m\in \M$. We write $\Im(f)$
    for the intermediate object,
    $\twoheadrightarrow$ for morphisms 
    $e\in \E$, and $\monoto$ for morphisms $m\in \M$.
    \end{minipage}
    \hfill
      \begin{tikzcd}[baseline=(A.base),row sep=2mm,column sep=4mm]
        |[alias=A]|
        A
        \arrow[->>,to path={|- (\tikztotarget) \tikztonodes},rounded corners]{dr}[pos=0.2]{e}
        \arrow[]{rr}[alias=f]{f}
        && B
        \\
        & \Im (f)
        \arrow[>->,to path={-| (\tikztotarget) \tikztonodes},rounded corners]{ur}[pos=0.8]{m}
      \end{tikzcd}
      \hfill
      \hspace*{0pt}

  \item\label{diagonalization}
    \begin{minipage}[t]{.7\textwidth}
      For each commutative square $g\cdot e = m\cdot f$ with $m\in \M$ and $e\in \E$,
      there exists a unique diagonal fill-in $d$ with $m\cdot d=g$ and $d\cdot e =
      f$.
    \end{minipage}
    \hfill
    \begin{tikzcd}[baseline=(A.base)]
      |[alias=A]|
      A
      \arrow[->>]{r}{e}
      \arrow{d}[swap]{f}
      & B
      \arrow{d}{g}
      \arrow[dashed]{dl}[description]{\exists !d}
      \\
      C
      \arrow[>->]{r}{m}
      & D
    \end{tikzcd}
    \hfill
    \hspace*{0pt}

  \end{enumerate}
\end{defiC}
\end{notheorembrackets}

\begin{example}
  In \Set, we have an $(\Epi,\Mono)$-factorization system
  where $\Epi$ is the class of surjective maps, and $\Mono$ the class of injective maps.
  The image of a map $f\colon A\to B$ is given by
  \[
    \Im(f) = \{b\in B \mid \text{there exists $a\in A$ with $f(a) = b$}\}.
  \]
  canonically yielding maps $e\colon A\twoheadrightarrow \Im(f)$ and $m\colon \Im(f)\monoto B$.
  Note that one can also regard $\Im(f)$ as a set of equivalence classes of $A$:
  \[
    \Im(f) \cong \big\{ \{a'\in A\mid f(a') = f(a)\} \,\mid a\in A \big\}.
  \]
  Intuitively, the diagonal fill-in property
  (\itemref{D:factSystem}{diagonalization}, also called \emph{diagonal lifting})
  provides a way of defining a map $d$ on equivalence classes (given by the
  surjective map at the top) and with a restricted codomain (given by the
  injective map at the bottom).
\end{example}

\begin{example}
  \label{trivFact}
  In general, the elements of $\E$ are not necessarily epimorphisms and the
  elements of $\M$ are not necessarily monomorphisms. In particular, every
  category has an $(\E,\M)$-factorization system with $\E := \Iso$ being the class
  of isomorphisms and $\M := \Mor$ being the class of all morphisms (and also
  vice-versa).
\end{example}

\begin{definition}
  An $(\E,\M)$-factorization system is called \emph{proper} if $\E\subseteq
  \Epi$ and $\M\subseteq \Mono$.
\end{definition}
These two conditions of properness are independent. In fact, $\M\subseteq \Mono$ is
equivalent to every split-epimorphism being in $\E$~\cite[Prop.~14.11]{joyofcats}.
In the literature, it is often required that the factorization system is proper,
and in fact a proper factorization system arises in complete or cocomplete
categories:

\begin{example} \label{completeCatFactorizations} Every complete category has
  a $(\StrongEpi, \Mono)$-factorization system \cite[Thm.~14.17 and
  14C(d)]{joyofcats} and also an $(\Epi, \StrongMono)$-factorization system
  \cite[Thm.~14.19, and 14C(f)]{joyofcats}. By duality, every cocomplete
  category has so as well.
\end{example}

\begin{remark} \label{rem:EM}
  $(\E,\M)$-factorization systems have many properties known from surjective
  and injective maps on $\Set$~\cite[Chp.~14]{joyofcats}:
  \begin{enumerate}
  \item\label{rem:EM:iso} $\E\cap \M$ is the class of isomorphisms of $\C$.
  \item\label{rem:EM:3} If $f\cdot g\in \M$ and $f\in \M$, then $g\in \M$. 
    If $\M\subseteq \Mono$, then $f\cdot g\in \M$ implies $g\in \M$.
  \item $\E$ and $\M$ are respectively closed under composition.
  \item\label{rem:EM:pullback} $\M$ is stable under pullbacks, $\E$ is stable under pushouts.
  \end{enumerate}
\end{remark}

The stability generalizes as follows to wide pullbacks and pushouts:
\begin{lemma}
  \label{EM:pullbackwide}
  $\M$ is stable under wide pullbacks: for a
    family $(f_i\colon A_i\to B)_{i\in I}$ and its wide pullback $(\pr_i\colon
    P\to A_i)_{i\in I}$, a projection $\pr_j\colon P\to A_j$ is in $\M$ if $f_i$
    is in $\M$ for all $i \in I\setminus \{j\}$.
\end{lemma}

A factorization system also provides notions of subobjects and quotients,
generalizing the notions of subset and quotient sets:

\begin{definition}
  \label{D:subobjectEM}
  For a class $\M$ of morphisms, an \emph{$\M$-subobject} of an object $X$ is 
  a pair $(S,s)$ where
  $s\colon S\monoto X$ is in $\M$.
  Two $\M$-subobjects $(s,S)$, $(s',S')$ are called \emph{isomorphic}
  if there is an isomorphism $\phi\colon S\to S'$ with $\phi\cdot s = s'$.
  We write $(s,S)\le (s', S')$ if there is a morphism $h\colon S\to S'$ with $s'\cdot h = s$.
  Dually, an \emph{$\E$-quotient} of $X$ is pair $(Q,q)$ for a
  morphism $q\colon X\epito Q$ ($q\in \E$).
  If $(\E,\M)$ is fixed from the context, we simply speak of
  subobjects and quotients.
\end{definition}
Note that $\le$ is not necessarily anti-symmetric: if $(s,S)\le (s',S')$ and $(s',S')\le (s,S)$
then it is not necessarily the case that $(s,S)$ and $(s',S')$ are isomorphic $\M$-subobjects.
Thus, it is often required that $\M$ is a class of
monomorphisms~\cite[Def.~7.77]{joyofcats}, but many of the results in the
present work hold without this assumption.
If $\M$ is so, then the subobjects (up to iso) of a given object $X$ form a preordered class. Moreover, they form a
preordered set iff $\C$ is $\M$-wellpowered. This is in fact
the definition: $\C$ is $\M$-wellpowered if for each $X\in \C$ there is (up to
isomorphism) only a set of $\M$-subobjects. On $\Set$, the isomorphism classes
of \text{($\Mono$-)}sub\-ob\-jects of $X$ correspond to subsets of $X$ and the
isomorphism classes of ($\Epi$-)quotients of $X$ correspond to partitions of
$X$.

If $(\E,\M)$ forms a factorization system, then its axioms provide us with
methods to construct and work with subobjects and quotients, e.g.~the image
factorization means that for every morphism, we obtain a quotient of its domain
and a subobject of its codomain. The minimization of coalgebras amounts to the
construction of certain subobjects or quotients with
respect to a suitable factorization system in the category of coalgebras $\Coalg(F)$.

\section{Factorization System for Coalgebras}
\label{secCoalgFact}

If we have an $(\E,\M)$-factorization system on the base category $\C$ on which
we consider coalgebras for $F\colon \C\to\C$, then it is natural
to consider coalgebra morphisms whose underlying $\C$-morphism is in $\E$, resp.~$\M$:
\begin{definition}
  Given a class of $\C$-morphisms $\E$, we say that an $F$-coalgebra morphism
  $h\colon (C,c)\to (D,d)$ is \emph{$\E$-carried} if $h\colon C\to D$ is in $\E$.
\end{definition}

This induces the standard notions of subcoalgebra and quotient coalgebras as
instances of $\M$-subobjects and $\E$-quotients in $\Coalg(F)$: an
$\M$-subcoalgebra of $(C,c)$ is an ($\M$-carried)-subobject of $(C,c)$ (in
$\Coalg(F)$), i.e.\ is represented by an $\M$-carried homomorphism $m\colon (S,s)
\monoto (C,c)$.
Likewise, a \emph{quotient}\index{quotient coalgebra} of a coalgebra $(C,c)$ is
an ($\E$-carried)-quotient of $(C,c)$ (in $\Coalg(F)$), i.e.\
is represented by a coalgebra morphism $e\colon (C,c) \epito (Q,q)$ carried by an
$\E$-morphism. If $\E$ happens to be a class of epimorphisms, then $q$ is uniquely determined by $e$ and $c$.

Note that for the case where $\M$ is the class of monomorphisms, the
monomorphisms in $\Coalg(F)$ coincide with the $\Mono$-carried homomorphisms
only under additional assumptions:
\begin{lemma}
  \label{coalgMono}
  If weak kernel pairs exist in $\C$ and are preserved by $F\colon \C\to \C$, then the
  monomorphisms in $\Coalg(F)$ are precisely the $\Mono$-carried coalgebra
  homomorphisms.
\end{lemma}
Preservation of kernel pairs is a commonly known criterion, and Gumm and
Schröder \cite[Example 3.5]{GS05} present an example of a functor that does
not preserve kernel pairs but for which there is a monic coalgebra
homomorphism that is not carried by a monomorphism.
\begin{proofappendix}{coalgMono}
  It is clear that every $\Mono$-carried homomorphism is monic in $\Coalg(F)$.
  Conversely, let $m\colon (C,c)\to (D,d)$ be a monomorphism in $\Coalg(F)$. Let
  $\pr_1,\pr_2\colon K\to C$ be a weak kernel pair of $m$. Since $F$ preserves
  weak kernel pairs, $F\pr_1,F\pr_2\colon FK\to FC$ is a weak kernel pair of
  $Fm\colon FC\to FD$. This induces some cone morphism $k\colon K\to FK$ making
  $\pr_1$ and $\pr_2$ coalgebra morphisms $(K,k)\to (C,c)$:
  \[
    \begin{tikzcd}
      K
      \arrow[shift left=1]{r}{\pr_1}
      \arrow[shift right=1]{r}[swap]{\pr_2}
      \arrow[dashed]{d}{k}
      & C
      \arrow{d}{c}
      \arrow{r}{m}
      & D
      \arrow{d}{d}
      \\
      FK
      \arrow[shift left=1]{r}{F\pr_1}
      \arrow[shift right=1]{r}[swap]{F\pr_2}
      & FC
      \arrow{r}{Fm}
      & FD
    \end{tikzcd}
  \]
  Since $m$ is monic in $\Coalg(F)$, this implies that $\pr_1=\pr_2$. For the verification that $m$ is a monomorphism in $\C$, consider $f,g\colon X\to C$ with $m\cdot f=m\cdot g$. Since $\pr_1,\pr_2$ is a weak kernel pair, it induces some cone morphism $v\colon X\to K$, fulfilling $f= \pr_1\cdot v$ and $g=\pr_2\cdot v$. Since, $\pr_1=\pr_2$, we find $f=g$
  as desired.
  \qed
\end{proofappendix}

For the construction of quotient coalgebras and subcoalgebras, it is handy to
have the factorization system directly in $\Coalg(F)$. It is a standard
result that the image factorization of homomorphisms lifts (see e.g.~\cite[Lemma
2.5]{lffiac}). Under assumptions on $\E$ and $\M$, Kurz shows that the
factorization system lifts to \Coalg(F)~\cite[Theorem 1.3.7]{kurzPhd} (and to
other categories with a forgetful functor to the base category $\C$).

In fact, the factorization system always lifts to $\Coalg(F)$ under the
condition that $F$ preserves $\M$. By this condition, we mean that $m\in \M$
implies $Fm\in \M$.
\begin{lemma}
  \label{coalgfactor}
  If $F\colon \C\to\C$ preserves $\M$, then the $(\E,\M)$-factorization system
  lifts from $\C$ to an ($\E$-carried, $\M$-carried)-factorization system in
  $\Coalg(F)$. The factorization of $F$-coalgebra homomorphisms
  and the diagonal fill-in morphisms in $\Coalg(F)$ are as in $\C$.
\end{lemma}
\begin{proof}
  We verify \autoref{D:factSystem}:
  \begin{enumerate}
  \item The $\E$- and $\M$-carried morphisms are closed under composition
    with isomorphisms, respectively.

  \item
    Given an $F$-coalgebra morphism $f\colon (A,a)\to (B,b)$,
  consider its factorization $f=m\cdot e$ in $\C$.
  Since $F$ preserves $\M$, we have $Fm\in \M$ and thus can apply the
  diagonal fill-in property (\itemref{D:factSystem}{diagonalization})
  to the coalgebra morphism square of $f$:
  \[
    \begin{tikzcd}[baseline={([yshift=-2mm]f.base)}]
      |[alias=A]|
      A
      \arrow[shiftarr={yshift=6mm}]{rr}[alias=f]{f}
      \arrow[->>]{r}{e}
      \arrow{d}{a}
      & \Im(f)
      \arrow[>->]{r}{m}
      \arrow[dashed]{d}{\exists! d}
      & B
      \arrow{d}{b}
      \\
      FA
      \arrow[->]{r}{Fe}
      \arrow[shiftarr={yshift=-6mm}]{rr}[swap]{Ff}
      & F\Im(f)
      \arrow[>->]{r}{Fm}
      & FB
    \end{tikzcd}
  \]
  This defines a unique
  coalgebra structure $d$ on $\Im(f)$ making $e$ and $m$ coalgebra morphisms.

  \item 
    In order to check the diagonal-lifting property of the
    ($\E$-carried, $\M$-carried)-factorization system, consider a commutative
    square $g\cdot e = m\cdot f$ in $\Coalg(F)$ with $m\in \M$, $e\in \E$:
    \[
      \begin{tikzcd}
        |[alias=A]|
        (A,a)
        \arrow[->>]{r}{e}
        \arrow{d}[swap]{f}
        & (B,b)
        \arrow{d}{g}
        \\
        (C,c)
        \arrow[>->]{r}{m}
        & (D,d)
      \end{tikzcd}
      \text{(in \(\Coalg(F)\))}
      \qquad\Longrightarrow\qquad
      \begin{tikzcd}[sep = 9mm]
        |[alias=A]|
        A
        \arrow[->>]{r}{e}
        \arrow{d}[swap]{f}
        & B
        \arrow{d}{g}
        \arrow[dashed]{dl}[description]{\exists !h}
        \\
        C
        \arrow[>->]{r}{m}
        & D
      \end{tikzcd}
      \text{ (in \(\C\))}
    \]
  In $\C$, there exists a unique $h\colon B\to C$ with
  $h\cdot e = f$ and $m\cdot h =g$.
  We only need to prove that $h\colon B\to C$ is a coalgebra homomorphism
  $(B,b)\to (C,c)$, i.e.~that $c\cdot h = Fh\cdot b$.
  We prove this equality by
  showing that both $c\cdot h$ and $Fh\cdot b$ are diagonals in a commutative
  square of the form of \itemref{D:factSystem}{diagonalization}.
  Indeed, we have the commutative squares:
  \[
    \begin{tikzcd}[sep=12mm]
      A
      \arrow{d}[swap]{a}
      \arrow{dr}{f}
      \arrow[->>]{rr}{e}
      &
      & B
      \arrow{dl}[swap]{h}
      \arrow{d}{g}
      \arrow{dr}{b}
      \\
      FA
      \arrow{d}[swap]{Ff}
      \descto{r}{f\text{ hom.}}
      & C
      \arrow[>->]{r}{m}
      \arrow{dl}{c}
      \descto{dr}{m\text{ hom.}}
      & D
      \descto{r}{g\text{ hom.}}
      \arrow{d}{d}
      & FB
      \arrow{dl}{Fg}
      \\
      FC
      \arrow[>->]{rr}{Fm}
      & & FD
    \end{tikzcd}
    \quad\text{and}\quad
    \begin{tikzcd}[sep=12mm]
      A
      \arrow[->>]{rr}{e}
      \arrow{d}[swap]{a}
      \descto{dr}{e\text{ hom.}}
      && B
      \arrow{dl}{b}
      \arrow{d}{b}
      \\
      FA
      \arrow{d}[swap]{Ff}
      \arrow{r}{Fe}
      & FB
      \arrow{dl}{Fh}
      \arrow{dr}{Fg}
      \descto{r}{\text{trivial}}
      & FB
      \arrow{d}{Fg}
      \\
      FC
      \arrow[>->]{rr}{Fm}
      & & FD
    \end{tikzcd}
  \]
  By the uniqueness of the diagonal in \itemref{D:factSystem}{diagonalization}, $c\cdot h =
  Fh\cdot b$.
  \qedhere
  \end{enumerate}
\end{proof}

\begin{remark}
  The condition that $F$ preserves $\M$ is commonly met. For $\Set$ and $\M$
  being the class of injective maps, it can be assumed wlog for coalgebraic
  purposes that $F$ preserves injective maps: every set functor preserves
  injective maps with non-empty domain and only needs to be modified on
  $\emptyset$ in order to preserve all injective maps~\cite{trnkova71}.
  The resulting functor has an isomorphic category of coalgebras.
\end{remark}

\begin{example}
  We saw an example $2\times(-)^{\set{a,b}}$-coalgebra morphism $h\colon
  (C,c)\to(D,d)$ between DFAs in \autoref{figExDfaHom}. Since $h$ was neither
  injective nor surjective, it properly factors into a surjective $e\colon
  (C,c)\epito (I,i)$ and an injective $m\colon (\Im(h),i)\monoto (D,d)$, $h=m\cdot e$, as
  illustrated in \autoref{figExDfaHomFactor}. The carrier of the intermediate
  coalgebra is just the image $\Im(h)$ of the map $h\colon C\to D$.
\end{example}

\begin{figure}
    \begin{tikzpicture}[coalgfit/.style={
        coalgebra frame,
        inner xsep=5mm,
        inner ysep=5mm,
        xshift=0mm,
        yshift=1mm,
      },
      x=13mm,y=13mm,
      ]
    \begin{scope}[coalgebra,shift={(0,0)},add state borders]
      \drawDfaC
      \node[coalgfit,fit={(s) (q) (p)}] (C) {};
      \node[coalgname,anchor=west] at (C.north west) {$(C,c)$};
    \end{scope}
    \begin{scope}[coalgebra,shift={(5cm,0)},add state borders]
      \drawDfaImh
      \node[coalgfit,inner xsep=3mm,xshift=1.5mm,fit={(s) (r) (p)}] (Imh) {};
      \node[coalgname,anchor=west] at (Imh.north west) {$(\Im(h),i)$};
    \end{scope}
    \begin{scope}[coalgebra,shift={(10cm,0)},add state borders]
      \drawDfaD
      \node[coalgfit,fit={(s) (r) (p)}] (D) {};
      \node[coalgname,anchor=west] at (D.north west) {$(D,d)$};
    \end{scope}
    \begin{scope}[path/.style={
        commutative diagrams/.cd, every arrow, every label,
        shorten <= 2mm,shorten >= 2mm,
      }]
      \path[path] (C) edge[->>]
      node[font=\normalsize] {$e$}
      (Imh) ;
      \path[path] (Imh) edge[>->]
      node[font=\normalsize] {$m$}
      (D) ;
      \coordinate (hStart) at ([xshift=2mm]C.north);
      \coordinate (hEnd) at ([xshift=3mm]D.north);
      \path[path,rounded corners] (hStart)
      -- ([yshift=8mm]hStart)
      -- node[font=\normalsize] {$h$} ([yshift=8mm]hStart -| hEnd)
      -- (hEnd) ;
    \end{scope}
  \end{tikzpicture}
  \caption{Factorization $h=m\cdot e$ in $\Coalg(2\times(-)^{\{a,b\}})$ of $h$
    from \autoref{figExDfaHom}}
  \label{figExDfaHomFactor}
\end{figure}
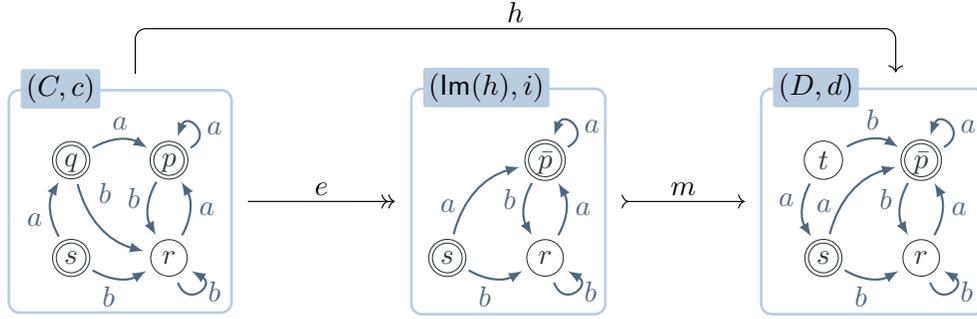

We have the dual result for factorization systems in $F$-algebras:
\begin{lemma}
  \label{algfactor}
  If $F\colon \C\to\C$ preserves $\E$, then the $(\E,\M)$-factorization system
  lifts from $\C$ to $\Alg(F)$.
\end{lemma}
\begin{proof}
  We have an $(\M,\E)$-factorization system in $\C^\dual$. By
  \autoref{coalgfactor}, this factorization system lifts to $\Coalg(F^\dual)$
  since $F^\dual\colon \C^\dual\to \C^\dual$ preserves $\E$. Thus, we have an
  ($\E$-carried, $\M$-carried)-factorization system in $\Alg(F) = \Coalg(F^\dual)^\dual$.
\end{proof}

This lifting result even holds for Eilenberg-Moore algebras for a monad $T\colon
\C\to\C$. The Eilenberg-Moore category of a monad $T$ is a full subcategory of
$\Alg(T)$ containing those algebras that interact coherently with the
multiplication $\mu\colon TT\to T$ and unit $\eta\colon \Id_\C\to T$
of the monad $T$, see~e.g.~Awodey~\cite{awodey2010category} for details.
Concretely, a $T$-algebra $(A,a)$ is an Eilenberg-Moore algebra if $a\cdot
\eta_A = \id_A$ holds and $a\colon TA\to A$ is a $T$-algebra homomorphism $a\colon
(TA,\mu_A)\to (A,a)$.
\begin{proposition}
  \label{monadfactor}
  If a monad $T\colon \C\to\C$ preserves $\E$, then the $(\E,\M)$-factorization system
  lifts from $\C$ to $\EM(T)$, the Eilenberg-Moore category of $T$.
\end{proposition}

\begin{proof}
  Denote the unit and multiplication of the monad $T$ by $\eta\colon \Id\to T$
  and $\mu\colon TT\to T$, respectively.
  Consider an $T$-algebra homomorphism $f\colon (A,a)\to (B,b)$ for Eilenberg-Moore
  algebras $(A,a)$ and $(B,b)$ and denote its image factorization in $\Alg(T)$
  by $(I,i)$, with homomorphisms $e\colon (A,a)\epito (I,i)$ and $m\colon
  (I,i)\monoto (B,b)$. We verify that $i\colon TI\to I$ is an Eilenberg-Moore algebra.
  \begin{itemize}
  \item First, we verify $i\cdot \eta_I = \id_I$ by showing that both $i\cdot
    \eta_I$ and $\id_I$ are both diagonals of the following square:

    
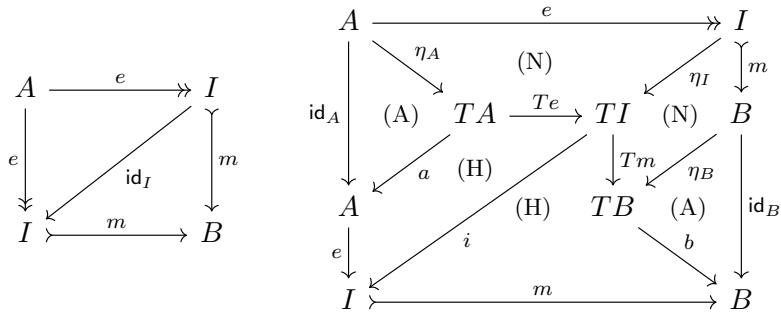
\[
      \begin{tikzcd}
        A
        \arrow[->>]{rr}{e}
        \arrow[->>]{dd}[swap]{e}
        & & I
        \arrow[>->]{dd}{m}
        \arrow[->]{ddll}{\id_I}
        \\
        \\
        I
        \arrow[>->]{rr}{m}
        && B
      \end{tikzcd}
      \qquad
      \begin{tikzcd}
        A
        \arrow[->>]{rrr}{e}
        \arrow{dr}{\eta_A}
        \arrow{dd}[swap]{\id_A}
        & {} \descto{dr}{(N)}
        & & I
        \arrow{dl}{\eta_I}
        \arrow[>->]{d}{m}
        \\
        {} \descto{r}{(A)}
        & TA
        \arrow{r}{Te}
        \arrow{dl}{a}
        & TI
        \descto{r}{(N)}
        \arrow{d}{Tm}
        \arrow{dldl}[pos=0.60]{i}
        & B
        \arrow{dl}{\eta_B}
        \arrow{dd}{\id_B}
        \\
        A
        \arrow{d}[swap]{e}
        &
        \descto{u}{(H)}
        \descto{r}{(H)}
        & TB
        \arrow{dr}{b}
        & \descto{l}{(A)}
        \\
        I
        \arrow[>->]{rrr}{m}
        & & & B
      \end{tikzcd}
    \]
    The left-hand square commutes trivially, and the right-hand square commutes
    because $\eta$ is natural (N), because $e$ and $m$ are $T$-algebra
    homomorphisms (H), and because $(A,a)$ and $(B,b)$ are Eilenberg-Moore
    algebras (A).
    By the uniqueness of the diagonal lifting property, we obtain $i\cdot \eta_I
    = \id_I$.

  \item Next, we verify $i\cdot Ti = i\cdot \mu_i$. In $\Alg(T)$, we have the
    following commutative square:
    \[
      \begin{tikzcd}
        (TA,\mu_A)
        \arrow[->>]{r}{Te}
        \arrow{d}[swap]{a}
        & (TI,\mu_I)
        \arrow{d}{Tm}
        \\
        (A,a) \arrow{d}[swap]{e}
        & (TB,\mu_B)
        \arrow{d}{b}
        \\
        (I,i)
        \arrow[>->]{r}{m}
        & (B,b)
      \end{tikzcd}
    \]
    All mentioned morphisms are indeed $T$-algebra morphisms, because $a$ and
    $b$ are Eilenberg-Moore algebras. By the previous lifting result
    (\autoref{algfactor}), the diagonal fill-in $(TI,i)\to (I,i)$ in $\Alg(T)$
    is given by the diagonal fill-in in $\C$, which is $i\colon TI\to I$. Hence,
    $i$ is a $T$-algebra homomorphism.
  \end{itemize}
  Thus, $(I,i)$ fulfils the axioms of an Eilenberg-Moore algebra. The remaining
  properties of the factorization system hold because the Eilenberg-Moore
  category is a full subcategory of $\Alg(T)$.
\end{proof}

The factorization system also lifts further to pointed coalgebras:

\begin{lemma}
  If $F\colon \C\to\C$ preserves $\M$, then the $(\E,\M)$-factorization system
  lifts from $\C$ to $\Coalg_I(F)$.
\end{lemma}
\begin{proof}
  A combination of \autoref{coalgfactor} and \ref{algfactor}, using that the
  constant functor preserves $\E$-morphisms.
  \takeout{}
\end{proof}

\section{Minimality in a Category}
\label{secMinObj}
Having seen multiple categories with an $(\E,\M)$-factorization system, we can now
define the minimality of objects abstractly.
\begin{definition}
  Given a category $\K$ with an $(\E,\M)$-factorization system, an object $C$ of
  $\K$ is called \emph{$\M$-minimal}\index{minimal@$\M$-minimal} if every
  morphism $h\colon D\monoto C$ in $\M$ is an isomorphism.
\end{definition}
\begin{remark}
  Every $(\E,\M$)-factorization system on $\K$ is an $(\M,\E)$-factorization
  system on $\K^\dual$, and thus induces a dual notion of
  $\E$-minimality: an object $C$ of $\K$ is called \emph{$\E$-minimal} if every $h\colon
  C\epito D$ in $\E$ is an isomorphism.
\end{remark}

In the following, $\K$ will denote
the category in which we consider the minimal objects, e.g.~the category of
coalgebras for a functor $F\colon \C\to \C$.
\begin{assumption}
  \label{assDualFactor}
  In the following, assume that the category $\K$ has an $(\E,\M)$-fac\-tor\-iza\-tion
  system. Whenever we consider a category of coalgebras for a functor $F\colon
  \C\to \C$ in the following, we achieve this by assuming that $\C$ has an
  $(\E,\M)$-factorization system and that $F$ preserves~$\M$.
\end{assumption}
The leading examples of the minimality notion in the present work are the following
two instances in coalgebras:
\begin{instance}
  For $\K := \Coalg_I(F)$, the ($\M$-carried-)minimal objects are the \emph{reachable}
  coalgebras, as introduced by~\Adamek\ \etal~\cite{amms13}.
  Concretely, an $I$-pointed $F$-coalgebra $(C,c,i_C)$ is reachable if it has no (proper)
  pointed subcoalgebra, equivalently, if every $\M$-carried coalgebra morphism
  $m\colon (S,s,i_S)\monoto (C,c,i_C)$ is necessarily an isomorphism~\cite{amms13}.

  In \Set, this corresponds to the usual notion of reachability: if a state
  $x\in C$ is contained in a subcoalgebra $m\colon (S,s,i_S)\monoto (C,c,i_C)$,
  then all successors of $x$ need to be contained in the subcoalgebra as well,
  since $m$ is a coalgebra homomorphism. Moreover, the subcoalgebra has to
  contain the point $i_C\colon I\to C$, and thus also all
  its successors, and in total all states reachable from $i_C$ in finitely many
  steps. Hence, $(C,c,i_C)$ is reachable if it is not possible to omit any
  state in a pointed subcoalgebra $(S,s,i_S)$, i.e.~if any such injective $m$ is
  a bijection.
\end{instance}

\begin{instance}
  For $\K := \Coalg(F)^\dual$, the ($\E$-carried-)minimal objects are called
  \emph{simple} coalgebras, as mentioned~by Gumm~\cite{Gumm03}. Usually, a
  simple coalgebra is defined as a coalgebra that does not have any proper
  quotient~\cite{WissmannDMS19}.\footnote{Gumm~\cite[p.~34]{Gumm03} defines a simple coalgebra as the
    quotient of a coalgebra on $\Set$ modulo behavioural equivalence.}

  In \Set, a coalgebra is simple iff all states have different behaviour --
  this characterization follows
  directly the following equivalent characterization of minimal objects as we
  will see in \autoref{instSimpleBehEq}:
\end{instance}

\begin{lemma}
  \label{minimalInE}
  An object $C$ in $\K$ is $\M$-minimal iff every $h\colon D\to C$ is in $\E$.
\end{lemma}
\begin{proof}
  In the \singlequote{if} direction, consider some $\M$-morphism $h\colon D\to
  C$. By the assumption, $h$ is also in $\E$ and thus an isomorphism.
  In the \singlequote{only if} direction, take some morphism $h\colon D\to C$
  and consider its $(\E,\M)$-factorization $e\colon D \epito \Im(h)$ and
  $m\colon \Im(h)\monoto C$ with $h=m\cdot e$. Since $C$ is $\M$-minimal, $m$ is
  an isomorphism and thus $h=m\cdot e$ is in $\E$.
\end{proof}
\begin{instance}
  \label{instSimpleBehEq}
  For $\K:=\Coalg(F)^\dual$, an $F$-coalgebra $(C,c)$ is simple iff every
  $F$-coalgebra morphism $h\colon (C,c)\to (D,d)$ is $\M$-carried.
\end{instance}

  In \Set, this equivalence shows that the simple coalgebras are precisely those
  coalgebras for which behavioural equivalence coincides with equality:
  \begin{itemize}
  \item If states $x,y\in C$ are behaviourally equivalent, then there is some
    $h\colon (C,c)\to (D,d)$ with $h(x)=h(y)$. By \autoref{instSimpleBehEq},
    $h$ must be injective and thus $x=y$.

  \item Conversely, if all states in $(C,c)$ have different behaviour, then
    every $h\colon (C,c)\to (D,d)$ is necessarily injective by
    \autoref{defBehEq}. Thus, by \autoref{instSimpleBehEq}, $(C,c)$ is simple.
  \end{itemize}

Gumm already noted that in $\Set$, every outgoing coalgebra morphism from a
simple coalgebra is injective~\cite[Hilfssatz~3.6.3]{Gumm03} -- but the converse direction (and thus the equivalence of \autoref{instSimpleBehEq}) was not mentioned. If $\E$, resp.~$\M$, happens to be the class of
epimorphisms, resp.~monomorphisms, yet another characterization of minimality exists:

\begin{lemma}
  \label{minimalSubterminal}
  Assume $\E=\Epi$ and weak equalizers in $\K$, then $X$ is $\M$-minimal iff
  there is at most one morphism $u\colon C\to D$ for every $D\in \K$.

  Dually, given $\M=\Mono$ and weak coequalizers, $C$ is $\E$-minimal iff $C$ is
  \emph{subterminal}, that is, iff there is a most one $u\colon D\to C$ for
  every $D\in \K$.
\end{lemma}
The name \emph{subterminal} stems from the fact that if $\K$ has a terminal
object, its subobjects are the subterminal objects.

\begin{proofappendix}{minimalSubterminal}
  We verify the postulated equivalence using \autoref{minimalInE} for $\E=\Epi$.
  \begin{itemize}
  \item For \textqt{if}, we verify that every $h\colon B\to C$ is an
    epimorphism: for $u,v\colon C\to D$ with $u\cdot h=v\cdot h$, we directly
    obtain $u=v$ by assumption. Thus, $h$ is an epimorphism.
  \item For \textqt{only if}, consider $u,v\colon  C\to D$ and take a weak
    equalizer $e\colon E\to C$; hence, $u\cdot e= v\cdot e$. Since $e$ is an
    epimorphism (by minimality), we obtain $u=v$.
    \qed
  \end{itemize}
\end{proofappendix}

\begin{instance}
  For $\K := \Coalg(F)^\dual$, assume $\M=\Mono$ and that $F$ preserves weak
  kernel pairs and that the base category $\C$ has coequalizers. Hence, the monomorphisms in $\Coalg(F)$ are precisely the
  $\Mono$-carried homomorphisms~(\autoref{coalgMono}) and the assumption of
  \autoref{minimalSubterminal} is met. Consequently, the simple
  coalgebras are precisely the \emph{subterminal} coalgebras. If the final
  coalgebra exists, then its subcoalgebras are precisely the simple coalgebras.
  For a non-example,
  Gumm and Schröder~\cite[Example 3.5]{GS05} provide a functor not 
  preserving weak kernel pairs and a subterminal coalgebra that is not simple.
\end{instance}

\begin{figure} \centering
  \begin{tabular}{@{}c@{\ }lc@{\hspace{-3mm}}c@{\,}l@{}}
    \cmidrule[\heavyrulewidth]{1-2}
    \cmidrule[\heavyrulewidth]{4-5}
    \multicolumn{2}{@{}l}{$X$ is $\M$-minimal}
    &
      \hspace*{5mm} %
    & \multicolumn{2}{@{}l}{$X$ is $\E$-minimal}
    \\
    \cmidrule{1-2} \cmidrule{4-5}
    \ensuremath{\Leftrightarrow} &
    every $Y\monoto X$ is an isomorphism
    && \ensuremath{\Leftrightarrow} &every $X\epito Y$ is an isomorphism
    \\
    \cmidrule{1-2} \cmidrule{4-5}
    \ensuremath{\Leftrightarrow} &
    every $Y\to X$ is in $\E$
    &&\ensuremath{\Leftrightarrow} & every $X\to Y$ is in $\M$
    \\
    \cmidrule{1-2} \cmidrule{4-5}
    &
    \makecell[l]{
      \textshaded{if $\E=\Epi$ and $\K$ has weak equalizers:} \\
    }
    &&& \makecell[l]{
       \textshaded{if $\M=\Mono$ and $\K$ has weak coequalizers:}
       }
    \\
    \ensuremath{\Leftrightarrow} &
      all parallel $X\rightrightarrows Y$ equal
    && \ensuremath{\Leftrightarrow} &
      all parallel $Y\rightrightarrows X$ equal
      (\textqt{X subterminal})
    \\
    \cmidrule[\heavyrulewidth]{1-2}
    \cmidrule[\heavyrulewidth]{4-5}
  \end{tabular}
  \caption{Equivalent characterizations of $\M$-minimality and
    $\E$-minimality in a category $\K$}
  \label{figCharacterizations}
\end{figure}

We have now established a series of equivalent characterizations of minimality
(\autoref{figCharacterizations}) and will now discuss how to construct minimal objects.
This process of minimization -- i.e.~of constructing the reachable part or the
simple quotient of a coalgebra -- is abstracted  as follows:

\begin{definition}
  An \emph{$\M$-minimization}\index{minimization@$\M$-minimization} of $C\in \K$
  is a morphism $m\colon D\monoto C$ in $\M$ where $D$ is $\M$-minimal.
\end{definition}

In fact, we will show in \autoref{uniqueMinimization} that an $\M$-minimization is unique, so we can speak of
\emph{the} $\M$-minimization.

\begin{instance}
  The task of finding an $\M$-minimization of a given $C\in \K$ instantiates to
  the standard minimization tasks on coalgebras:
  \begin{itemize}
  \item For $\K:=\Coalg_I(F)$, an $\M$-minimization of a given pointed coalgebra
    $(C,c,i_C)$ is called its \emph{reachable subcoalgebra}~\cite{amms13}. This is a
    subcoalgebra obtained by removing all unreachable states. The explicit
    definition is: the reachable subcoalgebra of $(C,c,i_C)$ is a (pointed)
    subcoalgebra $h\colon (R,r,i_R)\monoto (C,c,i_C)$ where $(R,r,i_R)$ itself has no proper
    (pointed) subcoalgebras.
  \item For $\K:= \Coalg(F)^\dual$, an $\E$-minimization of a given coalgebra
    $(C,c)$ is called the \emph{simple quotient} of $(C,c)$~\cite{Gumm03}. The
    explicit definition is: the simple quotient of $(C,c)$ is a quotient
    $h\colon (C,c)\epito (Q,q)$ where $(Q,q)$ itself has no proper quotient
    coalgebra.

    In \Set, this is a quotient in which all behaviourally equivalent states are
    identified, in other words, the simple quotient of $(C,c)$ is the unique
    coalgebra structure on $C/\mathord{\sim}$ that makes the canonical
    surjection $C\epito C/\mathord{\sim}$ a coalgebra morphism.
    Examples of simple quotients can
    be found in \autoref{figSimpleQuot}. Since all states in the codomain of the
    surjective homomorphisms are behaviourally inequivalent, the respective
    codomains are simple.
  \end{itemize}
\end{instance}

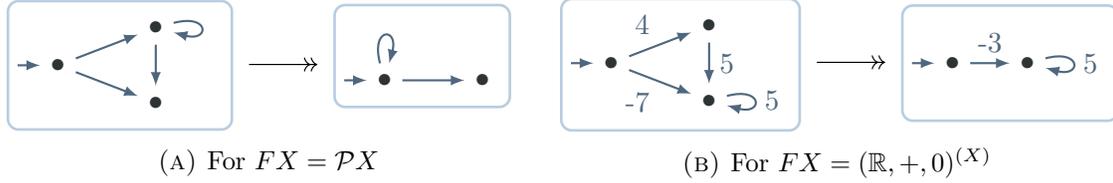
\begin{figure}[t]
  \hfill
  \begin{subfigure}{.46\textwidth}
    \begin{tikzpicture}[coalgebra,baseline=(q0.base)]
      \node[state] (q0) {$\bullet$};
      \node[state] (q1) at (1.3,0.5) {$\bullet$};
      \node[state] (q2) at (1.3,-0.5) {$\bullet$};
      \draw[transition] ([xshift=-3mm]q0.west) to (q0.west);
      \draw[transition] (q0) to (q1);
      \draw[transition] (q0) to (q2);
      \draw[transition] (q1) to (q2);
      \draw[transition,loop right] (q1) to node[alias=loopnode]{} (q1);
      \node[coalgebra frame,fit={([xshift=-3mm]q0.west) (q1) (q2) (loopnode)}] {};
    \end{tikzpicture}%
    \!\!%
    \begin{tikzcd}
      {} \arrow[->>]{r} & {}
    \end{tikzcd}%
    \!\!%
    \begin{tikzpicture}[coalgebra,baseline={([yshift=2mm]q0.base)}]
      \node[state] (q0) {$\bullet$};
      \node[state] (q1) at (1.3,0) {$\bullet$};
      \draw[transition] ([xshift=-3mm]q0.west) to (q0.west);
      \draw[transition] (q0) to (q1);
      \draw[transition,loop above] (q0) to node[alias=loopnode]{} (q0);
      \node[coalgebra frame,fit={([xshift=-3mm]q0.west) (q1) (loopnode)}] {};
    \end{tikzpicture}
    \subcaption{For $FX=\Pow X$}
  \end{subfigure}%
  \hfill
  \begin{subfigure}{.49\textwidth}
    \hfill
        \begin{tikzpicture}[coalgebra,baseline=(q0.base)]
          \node[state] (q0) {$\bullet$};
          \node[state] (q1) at (1.3,0.5) {$\bullet$};
          \node[state] (q2) at (1.3,-0.5) {$\bullet$};
          \draw[transition] ([xshift=-3mm]q0.west) to (q0.west);
          \draw[transition] (q0) to node[above left] {4}(q1);
          \draw[transition] (q0) to node[below left] {-7} (q2);
          \draw[transition] (q1) to node[right] {5}(q2);
          \draw[transition,loop right,overlay] (q2) to node[right,alias=loopnode] {5}(q2);
          \node[coalgebra frame,fit={([xshift=-3mm]q0.west) (q1) (q2) (loopnode)}] {};
        \end{tikzpicture}%
        \!\!%
        \begin{tikzcd}
          {} \arrow[->>]{r}{}
          & {}
        \end{tikzcd}\!\!%
        \begin{tikzpicture}[coalgebra,baseline=(q0.base)]
          \node[state] (q0) {$\bullet$};
          \node[state] (q1) at (1,0) {$\bullet$};
          \draw[transition] ([xshift=-3mm]q0.west) to (q0.west);
          \draw[transition] (q0) to node[above] {-3} (q1);
          \draw[transition,loop right] (q1) to node[right,alias=loopnode] {5}(q1);
          \node[coalgebra frame,inner ysep=5mm,fit={([xshift=-3mm]q0.west) (q1) (loopnode)}] {};
        \end{tikzpicture}%
    \subcaption{For $FX=(\R,+,0)^{(X)}$}
  \end{subfigure}
  \vspace{-2mm} %
  \caption{Examples of simple quotients in $F$-coalgebras}
  \label{figSimpleQuot}
\end{figure}
  
\begin{samepage}
\begin{example} \label{trivMinimal}
  For the trivial factorization systems (\autoref{trivFact}), we have:
  \begin{itemize}
  \item For the $(\Iso,\Mor)$-factorization system, the $\Iso$-minimization of
    an object $X$ is $X$ itself.
  \item For the $(\Mor,\Iso)$-factorization system on category, if a strict
    initial object $0$ exists, then it is the $\Mor$-minimization of every $X\in
    \C$. Recall that an initial object $0$ is called \emph{strict} if every
    morphism with codomain $0$ is an isomorphism.
  \end{itemize}
\end{example}
\end{samepage}

It is well-defined to speak of \emph{the} $\M$-minimization of an object $C$,
because it is unique:

\begin{lemma}\label{MminLeast}
  Consider $h\colon M\to C$ with $\M$-minimal $M$ and an $\M$-subobject
  $s\colon S\monoto C$. The pullback of $s$ along $h$ exists iff $h$ factors
  uniquely through $s$, that is, iff there is a unique $u\colon M\to S$ with $s\cdot u = h$.
\end{lemma}
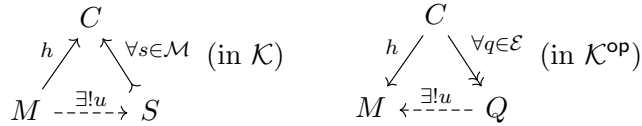
\[
  \begin{tikzcd}[column sep=5pt]
    &C
    \\
    M
    \arrow[->]{ur}{h}
    \arrow[->,dashed]{rr}{\exists!u}
    &&
    S
    \arrow[>->]{ul}[swap]{\forall s\in \M}
  \end{tikzcd}
  ~~\text{(in $\K$)}
  \qquad
  \begin{tikzcd}[column sep=5pt]
    &C
    \\
    M
    \arrow[<-]{ur}{h}
    \arrow[<-,dashed]{rr}{\exists!u}
    &&
    Q
    \arrow[<<-]{ul}[swap]{\forall q\in \E}
  \end{tikzcd}
  ~~\text{(in $\K^\dual$)}
\]
\begin{proof}
  In the `if'-direction, let $d\colon M\to S$ be the unique morphism with
  $s\cdot d = h$. The pullback is simply given by $M$ itself with
  projections $\id_M\colon M\to M$ and $d\colon M\to S$.
  To verify its universal
  property, consider $e\colon E\to M$, $f\colon E\to S$ with $h\cdot e = s\cdot
  f$ (\autoref{fig:MminIf}).
  Since $M$ is $\M$-minimal, $e\colon E\to M$ is in $\E$ (\autoref{minimalInE}).
  Thus, we can apply the diagonal lifting property to $h\cdot e = s\cdot f$
  yielding a diagonal $u$ with $s\cdot u=h$ and $u\cdot e = f$.
  Thus, $d = u$ and $d\cdot e = f$, showing that $e\colon (E,e,f)\to (M,\id_M,d)$ is the
  mediating cone morphism. Its uniqueness is clear because $\id_M$ is an isomorphism.
  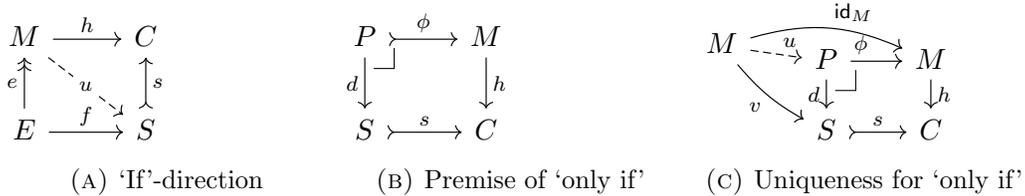
\begin{figure}[h!]
    \begin{subfigure}[b]{.3\textwidth}
      \begin{tikzcd}
        M
        \arrow{r}{h}
        \arrow[dashed]{dr}[description]{u}
        & C
        \\
        E
        \arrow[->>]{u}{e}
        \arrow{r}{f}
        & S
        \arrow[>->]{u}[swap]{s}
      \end{tikzcd}
      \caption{`If'-direction}
      \label{fig:MminIf}
    \end{subfigure}%
    \begin{subfigure}[b]{.3\textwidth}%
      \begin{tikzcd}
        P
        \arrow[>->]{r}{\phi}
        \arrow[->]{d}[swap]{d}
        \pullbackangle{-45}
        &
        M
        \arrow[->]{d}{h}
        \\
        S
        \arrow[>->]{r}{s}
        & C
      \end{tikzcd}
      \caption{Premise of `only if'}
      \label{fig:MminPullback}
    \end{subfigure}
    \begin{subfigure}[b]{.3\textwidth}
      \begin{tikzcd}[row sep=0mm,column sep=7mm]
        |[yshift=2mm]|
        M
        \arrow[bend left=20]{rr}{\id_M}
        \arrow[bend right=10]{dr}[swap]{v}
        \arrow[dashed]{r}{u}
        &
        P
        \arrow{r}[pos=0.2]{\phi}
        \arrow{d}[swap]{d}
        \pullbackangle{-45}
        &
        M
        \arrow[->]{d}{h}
        \\[4mm]
        & S
        \arrow[>->]{r}{s}
        & C
      \end{tikzcd}
      \caption{Uniqueness for `only if'}
      \label{fig:MminUMP}
    \end{subfigure}
    \caption{Diagrams for the proof of \autoref{MminLeast}}
  \end{figure}
  
  In the `only if'-direction, consider the pullback $(P,\phi,d)$ (\autoref{fig:MminPullback}).
  Since $\M$-morphisms are stable under pullback
  (\itemref{rem:EM}{rem:EM:pullback}), $\phi$ is in
  $\M$, too. By the minimality of $M$, the $\M$-morphism $\phi$ is an isomorphism
  and we have $d\cdot \phi^{-1}\colon M\to S$.

  In order to see that $d\cdot \phi^{-1}$ is indeed the unique morphism $M\to S$ making the
  triangle commute, consider an arbitrary $v\colon M\to S$ with $s\cdot v = h =
  h\cdot \id_M$.
  Thus, $M$ is a competing cone for the pullback $P$ and thus induces a morphism
  $u\colon M\to P$ with $d\cdot u = v$ and $\phi\cdot u =\id_M$ (\autoref{fig:MminUMP}). Since $\phi$ is
  an isomorphism, we have $u=\phi^{-1}$ and thus $v=d\cdot \phi^{-1}$ as desired.
\end{proof}

In the following, we use the terminology \emph{$\M$-intersection} for a
pullback of an $\M$-morphism along another $\M$-morphism. For $\M$ being the
injective maps in $\Set$, the $\M$-intersection of two $\M$-subobjects boils
down to an ordinary intersection.
\begin{corollary}
  \label{uniqueMinimization}
  If all $\M$-intersections exist in $\K$, then
  an $\M$-minimization $m\colon M\monoto C$ is the least $\M$-subobject of $C$ (w.r.t.~the preorder $\le$) and unique up to unique isomorphism.
\end{corollary}%
\begin{proof}
  Consider \autoref{MminLeast} first for $h\in \M$ and then also with
  $S$ being $\M$-minimal.
\end{proof}
Concretely, for every $\M$-subobject $s\colon S\monoto C$, there is a unique
$u\colon M\to S$ (which is necessarily in $\M$) such that:
\[
  \begin{tikzcd}[column sep=5pt]
    &C
    \\
    M
    \arrow[>->]{ur}{m}
    \arrow[>->,dashed]{rr}{\exists!u}
    &&
    S
    \arrow[>->]{ul}[swap]{s}
  \end{tikzcd}
\]

\begin{instance} The results instantiate to the uniqueness results in the instances of reachable subcoalgebras and simple quotients.
  \begin{listinenv}
  \begin{enumerate}
  \item If $\C$ has pullbacks of $\M$-morphisms
    (i.e.~finite intersections) and $F\colon \C\to \C$ preserves them,
    then $\Coalg_I(F)$ has pullbacks of $\M$-carried homomorphisms. Given a reachable
    subcoalgebra $(D,d,i_D)$ of $(C,c,i_C)$, then it is the least $I$-pointed
    subcoalgebra of $(C,c,i_C)$ (cf.~\cite[Notation~3.18]{amms13}) and is unique
    up to isomorphism.
  \item If $\C$ has pushouts of $\E$-morphisms, then $\Coalg(F)$ has pushouts of
    $\E$-carried homomorphisms. Hence, the simple quotient of a coalgebra
    $(C,c)$ is the greatest quotient of $(C,c)$ and unique up to isomorphism
    (e.g.~\cite[Lemma~2.9]{WissmannDMS19}).
  \end{enumerate}
  \end{listinenv}
\end{instance}

There are instances where a minimization $M$ exists, but where a mediating
morphism in the sense of \autoref{MminLeast} is not unique:

\begin{example}[Tree unravelling]
  \label{treeUnravel}
  Let $\Coalg_I(F)_\reach$ be the category of reachable pointed $F$-coalgebras,
  i.e.~the full subcategory $\Coalg_I(F)_\reach\subseteq \Coalg_I(F)$ such that
  $(C,c,i_C) \in \Coalg_I(F)_\reach$ iff it is reachable. For simplicity,
  restrict to $F\colon \Set\to\Set$ with the $(\Epi,\Mono)$-factorization
  system. Thus, all morphisms in $\Coalg_I(F)_\reach$ are surjective
  (\autoref{minimalInE}). Considering the (trivial) $(\Mor,\Iso)$-factorization
  system on $\Coalg_I(F)_\reach$, a coalgebra $(C,c,i_C)$ is ($\Mor$-)minimal iff
  every coalgebra morphism $h\colon (D,d,i_D)\to (C,c,i_C)$ (with $(D,d,i_D)$
  also reachable) is an isomorphism. If $(C,c,i_C)$ is $\Mor$-minimal, then it
  is a tree: to see this, take $h$ to be its tree unravelling (see e.g.~\autoref{fig:unravel}),
  and by the $\Mor$-minimality, $h$ is an isomorphism, so $(C,c,i_C)$ is already a tree.
  \hackynewpage

  This implies that if the ($\Mor$-)minimization of a coalgebra exists, then it is
  its tree unravelling. For example, for $I=1$ and the bag functor $FX=\Bag X$,
  we have the minimizations as illustrated in \autoref{fig:unravel}. It is easy
  to see that for $FX=\Pow X$ however, no coalgebra (with at least one
  transition) has a $\Mor$-minimization, because one can always duplicate
  successor states.\footnote{%
    The $\Mor$-minimization of reachable $F$-coalgebras is related to the
    so-called $F$-precise factorizations \cite[Def.~3.1, 3.4]{WissmannDKH19}.}

  For $FX=\Bag X$, all $\Mor$-minimizations exist, but they are not unique up to
  unique isomorphism.
  Consider the tree unravelling $m\colon M\epito X$ in
  \autoref{fig:unravel:sib}. There is an isomorphism $\phi\colon M\to M$ that
  swaps the two successors of the initial state. Hence, $\phi\neq \id_B$, but
  $m\cdot \phi = m\cdot \id_M$, so $M$ is unique up to isomorphism, but not
  unique up to unique isomorphism.
\end{example}
\begin{figure}
  \hfill
  \begin{subfigure}{.5\textwidth}
        \begin{tikzpicture}[coalgebra,baseline=(q0.base)]
          \node[state] (q0) {$\bullet$};
          \node[state] (q1) at (1.3,0.5) {$\bullet$};
          \node[state] (q2) at (1.3,-0.5) {$\bullet$};
          \draw[transition] ([xshift=-3mm]q0.west) to (q0.west);
          \draw[transition] (q0) to node[above left] {}(q1);
          \draw[transition] (q0) to node[below left] {} (q2);
          \node[coalgebra frame,fit={([xshift=-3mm]q0.west) (q1) (q2)}] {};
        \end{tikzpicture}%
        \begin{tikzcd}
          {} \arrow[->>]{r} & {}
        \end{tikzcd}%
        \begin{tikzpicture}[coalgebra,baseline=(q0.base)]
          \node[state] (q0) {$\bullet$};
          \node[state] (q1) at (1.5,0) {$\bullet$};
          \draw[transition] ([xshift=-3mm]q0.west) to (q0.west);
          \draw[transition,bend left=20] (q0) to node[above] {}(q1);
          \draw[transition,bend right=20] (q0) to node[below] {} (q1);
          \node[coalgebra frame,inner ysep=4mm,fit={([xshift=-3mm]q0.west) (q1)}] {};
        \end{tikzpicture}
        \subcaption{Unravelling of siblings}
        \label{fig:unravel:sib}
  \end{subfigure}%
  \hfill
  \begin{subfigure}{.48\textwidth}
    \hfill
        \begin{tikzpicture}[coalgebra,baseline=(q0.base)]
          \node[state] (q0) {$\bullet$};
          \node[state] (q1) at (1,0) {$\bullet$};
          \node[state] (q2) at (2.3,0) {$\cdots$};
          \draw[transition] ([xshift=-3mm]q0.west) to (q0.west);
          \draw[transition] (q0) to node[above] {}(q1);
          \draw[transition] (q1) to node[above] {} (q2);
          \node[coalgebra frame,inner ysep=2mm,fit={([xshift=-3mm]q0.west) (q1) (q2)}] {};
        \end{tikzpicture}%
        \begin{tikzcd}
          {} \arrow[->>]{r} & {}
        \end{tikzcd}%
        \begin{tikzpicture}[coalgebra,baseline=(q0.base)]
          \node[state] (q0) {$\bullet$};
          \draw[transition] ([xshift=-3mm]q0.west) to (q0.west);
          \draw[transition,loop right,out=30,in=-30,looseness=10]
              (q0) to node[right,alias=loopnode] {}(q0);
              \node[coalgebra frame,inner ysep=2mm,fit={([xshift=-3mm]q0.west) (q0) (loopnode)}] {};
        \end{tikzpicture}
        \subcaption{Unravelling of a loop}
  \end{subfigure}
  \caption{Tree unravelling for $FX=\Bag X$}
  \label{fig:unravel}
\end{figure}
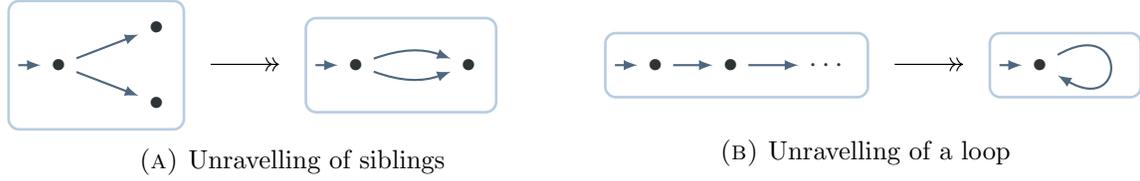

For proving the existence of an $\M$-minimization, we assume that $\M$
is a subclass of the monomorphisms in $\K$. Under this assumption, we first
establish the converse of \autoref{uniqueMinimization}:
\begin{lemma}
  \label{leastIsMinimal}
  If $\M\subseteq \Mono$ and if the least $\M$-subobject $M$ of $X$ exists, then
  $M$ is the $\M$-minimization of $X$.
\end{lemma}
\begin{proof}
  Let $m\colon M\monoto X$ be the least $\M$-subobject of $X$, and consider
  $s\colon S\monoto M$ in $\M$. Since $m\cdot s\in \M$, there is some $u\colon
  M\to S$ with $(m\cdot s)\cdot u = m$. Since $m$ is monic, we obtain $s\cdot u
  = \id_M$. Hence, $s$ is a split-epimorphism, and together with
  $s\in\M\subseteq \Mono$, $s$ is an isomorphism.
\end{proof}
\begin{proposition}
  \label{existenceMinimization}
  If $\M\subseteq \Mono$, $\K$ has wide pullbacks of
  $\M$-morphisms, and $\K$ is $\M$-wellpowered, then
  every object $C$ of $\K$ has an $\M$-minimization.
\end{proposition}
\begin{proof}
  Since $\K$ is $\M$-wellpowered, all the $\M$-carried morphisms $m\colon
  M\monoto C$ form up to isomorphism a set $S$. The wide pullback of all $m\in
  S$ exists in $\K$ by assumption, denote it by $\pr_m\colon P \to M$ for
  $m\colon M\monoto C$. All $m'\in S$ are in $\M$ and so are all $\pr_m$
  by \autoref{EM:pullbackwide}. Since, $\id_C\in \M$,
  there must be some $m\colon M\monoto C$ in $S$ such that $(m,M)$ and $(\id_C,C)$ are isomorphic $\M$-subobjects.
  Hence, $p := \id_m\cdot \pr_{m}\colon P\monoto C$ represents an
  $\M$-subobject, and moreover the least
  $\M$-subobject of $C$, as witnessed by the projections $\pr_m$. By
  \autoref{leastIsMinimal}, $P$ is the minimization of $C$.
\end{proof}
\begin{instance} \label{exExistenceCoalgMinimal}
  This proof directly instantiates to the proofs of the existence of the
  reachable subcoalgebra and simple quotient:
  \begin{listinenv}
    \begin{enumerate}
    \item In the reachability case, let $\M$ be a subclass of the monomorphisms, let
      the base category $\C$ have all (set-indexed) $\M$-intersections, and let $F\colon \C\to
      \C$ preserve all (set-indexed) intersections. Then the reachable part of a given pointed
      coalgebra $(C,c,i_C)$ is obtained as the intersection of all pointed
      subcoalgebras of $(C,c,i_C)$ \cite{amms13}.

      For $\C=\Set$, and $\M$ being the class of injective maps, all
      intersections exist. The condition that $F\colon \Set\to\Set$ preserves
      all intersections is mild: all finitary functors preserve all
      intersections~(\cite[Proof of Lem.~8.8]{amm18} or \cite[Lem.~2.6.10]{Wissmann2020}) and many non-finitary functors do
      as well, e.g.~the powerset functor. An example of a functor that does not
      preserve all intersections is the filter functor~\cite[Sect.~5.3]{Gumm2001}.

    \item For the existence of simple quotients, let $\E$ be a subclass of
      the epimorphisms and let the base category $\C$ be cocomplete and
      $\E$-cowellpowered. Then every $F$-coalgebra (C,c) has a simple quotient
      given by the wide pushout of all quotient coalgebras
      (\cite[Proposition 3.7]{amms13}, and \cite{Gumm08} for the instance $\C=\Set$).

      Every set has only a set of outgoing surjective maps, so all assumptions
      are met for $\C=\Set$, $\E$ containing only surjective maps, and every $\Set$-functor $F$.
    \end{enumerate}
  \end{listinenv}
\end{instance}

\begin{remark}
  All observations on simple quotients also apply to pointed coalgebras:
    An $I$-pointed $F$-coalgebra is simple iff it is
      $\E$-carried-minimal in $\K:=\Coalg_I(F)^\dual$. 
      The forgetful functor
      \[
        \Coalg_I(F) \longrightarrow \Coalg(F)
      \]
      preserves and reflects simple coalgebras and simple quotients (note that
      for every pointed coalgebra $(C,c,i_C)$, the slice categories
      $(C,c,i_C)/\Coalg_I(F)$ and $(C,c)/\Coalg(F)$ are isomorphic). For the sake
      of simplicity, we will not state the results explicitly for simple
      coalgebras in $\Coalg_I(F)$.
\end{remark}

\begin{definition}
  We denote by $J\colon \K_{\min}\hookto \K$ the full subcategory formed by the
  $\M$-minimal objects of $\K$.
\end{definition}

In the existence proof of minimal objects
(\autoref{existenceMinimization}) we only required (wide) pullbacks
where all morphisms in the diagram are in $\M$. We obtain additional properties
if we assume the pullback along $\M$-morphisms, i.e.~pullbacks where only one of
the two morphisms is in $\M$:

\begin{proposition}
  \label{MminCorefl}
  Suppose that pullbacks along $\M$-morphisms exist in $\K$ and that every
  object of $\K$ has an $\M$-minimization. Then $J\colon
  \K_{\min}\hookto \K$ is a coreflective subcategory. Its
  right-adjoint $R\colon \K\to \K_{\min}$ ($J\dashv R$) sends an object to its
  $\M$-minimization; in particular, minimization is functorial.
\end{proposition}

\begin{proof}
  The universal property of $J$ follows directly from \autoref{MminLeast}: To
  this end, it suffices to consider $R$ as an object assignment. Given a
  morphism $h\colon M\to X$ where $M$ is $\M$-minimal, we need to show that it
  factorizes uniquely through the $\M$-minimization $s\colon S\monoto X$ of
  $D$, that is $RD := S$. Since the pullback of $h$ along $s$ exists by
  assumption, \autoref{MminLeast} yields us the desired unique factorization
  $u\colon M\to S$ with $s\cdot u = h$.
\end{proof}

\begin{samepage}
\begin{instance} \label{exExistsMinimization}
  For both of our main instances, this adjunction has been observed before:
  \begin{listinenv}
    \begin{enumerate}
    \item If $F\colon \C\to \C$ preserves inverse images (w.r.t.~$\M$), then pullbacks along
      $\M$-carried homomorphisms exist in $\Coalg_I(F)$. Hence, the reachable
      $I$-pointed $F$-coalgebras form a coreflective subcategory of
      $\Coalg_I(F)$, where the coreflector maps a pointed coalgebra to its
      reachable part~\cite[Thm~5.23]{wmkd20reachability}

    \item\label{iExistsSimple} The simple coalgebras form a reflective
      subcategory of $\Coalg(F)$, and the reflector sends a coalgebra to its
      simple quotient, under the assumption that the base category has pushouts
      along $\E$-morphisms. For coalgebras in \Set, the adjunction $J\vdash R$ has been shown by
      Gumm~\cite[Theorem 2.3]{Gumm08}.
    \end{enumerate}
  \end{listinenv}
\end{instance}
\end{samepage}

\begin{corollary}\label{MminQuot}
  If pullbacks along $\M$-morphisms exist in $\K$ and all $\M$-minimizations
  exist, then $\M$-minimal objects are closed under $\E$-quotients.
\end{corollary}
\hackynewpage
\begin{proof}
  Consider an $\E$-morphism $e\colon C\epito D$ where $C$ is $\M$-minimal. Take
  the adjoint transpose $f\colon C\to RD$ with $m\cdot f = e$ where $m\colon
  RD\monoto D$ is the $\M$-minimization of $D$:
  \[
    \begin{tikzcd}
      C
      \arrow[->>]{dr}{e}
      \arrow{d}[swap]{f}
      \\
      RD
      \arrow[>->]{r}{m}
      & D
    \end{tikzcd}
  \]
  Since $RD$ is $\M$-minimal, $f$ is in $\E$ (\autoref{minimalInE}).
  Moreover, $f\in \E$ and $m\cdot f\in \E$ imply $m\in \E$
  (\itemref{rem:EM}{rem:EM:3}), hence $m\in \E\cap\M$ is an isomorphism and
  $D$ is $\M$-minimal.
\end{proof}

\begin{remark}
  The closure of $\M$-minimal objects under quotients also holds under slightly
  different assumptions. For example, closure can be shown
  \begin{enumerate}
  \item if pullbacks along
  $\M$-morphisms exist in $\K$ and $\M$ is a class of monomorphisms,
  \item or if $\E$-morphisms are closed under pullbacks and all those exist in $\K$.
  \end{enumerate}
\end{remark}

In the example of the factorization of a DFA-morphism
(\autoref{figExDfaHomFactor} on p.~\pageref{figExDfaHomFactor}), every state in
$(C,c)$ was reachable from $s$, and hence, every state in the quotient
$e\colon (C,c)\epito (\Im(h),i)$ is reachable from $e(s)$.

\begin{example}
  \label{exQuotClosed}
  \begin{listinenv}
    \begin{enumerate}
    \item\label{exQuotClosed:Reach} If $F$ preserves inverse images, then reachable $F$-coalgebras are
      closed under quotients \cite[Cor.~5.24]{wmkd20reachability}. Note that if
      $F$ does not preserve inverse images, then a quotient of a reachable
      $F$-coalgebra may not be reachable. For example, in (pointed) coalgebras
      for the monoid-valued functor $(\R,+,0)^{(-)}$ there is the coalgebra
      quotient with $h(b_1)=h(b_2) = b$:
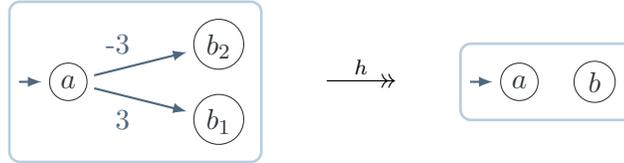
      \begin{center}
        \begin{tikzpicture}[coalgebra,baseline=(q0.base),add state borders]
          \node[state] (q0) {$a$};
          \node[state] (q1) at (2,-0.5) {$b_1$};
          \node[state] (q2) at (2,0.5) {$b_2$};
          \draw[transition] ([xshift=-3mm]q0.west) to (q0.west);
          \draw[transition] (q0) to node[below left] {3}(q1);
          \draw[transition] (q0) to node[above left] {-3} (q2);
          \node[coalgebra frame,fit={([xshift=-3mm]q0.west) (q0) (q1) (q2)}] {};
        \end{tikzpicture}
        \quad
        \begin{tikzcd}
          {} \arrow[->>]{r}{h}
          & {}
        \end{tikzcd}
        \quad
        \begin{tikzpicture}[coalgebra,baseline=(q0.base),add state borders]
          \node[state] (q0) {$a$};
          \node[state] (q1) at (1,0) {$b$};
          \draw[transition] ([xshift=-3mm]q0.west) to (q0.west);
          \node[coalgebra frame,fit={([xshift=-3mm]q0.west) (q0) (q1)}] {};
        \end{tikzpicture}
      \end{center}
      Since transition weights may cancel out each other ($-3 + 3 = 0$), the
      codomain of $h$ is not reachable even though its domain is.
    \item If the base category $\C$ has pushouts along $\E$-morphisms, then
      simple $F$-coalgebras are closed under subcoalgebras.
      For $\C=\Set$, this is obvious: if in a coalgebra $(C,c)$, all states are
      of pairwise different behaviour, then so they are in every subcoalgebra of
      $(C,c)$.
    \end{enumerate}
  \end{listinenv}%
\end{example}

\takeout{}

\section{Interplay of minimality notions}
\label{secInterplay}

The two main aspects of minimization we have seen -- reachability and minimization
for observability -- are closely connected on an abstract level 
and also interact well as we see in the following. In order to minimize a
pointed coalgebra under both aspects, we have two options: first construct the
reachable part and then the simple quotient, or we first form the simple
quotient and then construct its reachable part. Given the existence of pullbacks
of $\M$-morphisms along arbitrary morphisms,
we show in \autoref{propWellPointed} that any order is fine.

In the abstract setting of a category $\K$ with an $(\E,\M)$-factorization
system we are transforming an object $C\in \K$ into an object $C'$ that is
$\M$-minimal in $\K$ and $\E$-minimal in~$\K^\dual$.
\begin{proposition} \label{propWellPointed}
  Suppose $\K$ has an $(\E,\M)$-factorization system such that
  all $\M$-mi\-nim\-iza\-tions in $\K$ and all $\E$-minimizations in $\K^\dual$ exist.
  If $\K$ has pullbacks along $\M$-morphisms and pushouts along $\E$-morphisms,
  then for every $C$ in $\K$ the following two constructions yield the same object:
  \begin{enumerate}
  \item The $\M$-minimization of $C$ in $\K$ followed by its $\E$-minimization
    in $\K^\dual$.
  \item\label{simpleFirst} The $\E$-minimization of $C$ in $\K^\dual$ followed by its $\M$-minimization
    in $\K$.
  \end{enumerate}
\end{proposition}
\begin{proof}
  In the first approach,
  denote the $\M$-minimization of $C$ by $m\colon R\monoto C$ and its
  $\E$-minimization by $s\colon R \epito V$. In the other approach, denote the
  $\E$-minimization of $C$ by $e\colon C\epito Q$ and its $\M$-minimization by
  $t\colon W\monoto Q$:
  \[
    \begin{tikzcd}
      C
      \arrow[->>]{rr}{e}
      &
      &
      Q\\
      R
      \arrow[>->]{u}{m}
      \arrow[->>]{r}{s}
      & V
      & W
      \arrow[>->]{u}{t}
    \end{tikzcd}
  \]
  We need to prove that $V$ and $W$ are isomorphic, making the above
  (then-closed) square commute. The $\M$-minimal objects form a coreflective
  subcategory (\autoref{MminCorefl}), so $e\cdot m$, whose domain is
  $\M$-minimal, factorizes through the $\M$-minimization of the codomain of
  $e\cdot m$, i.e.~we have $h\colon R\to W$ with $t\cdot h = e\cdot m$. Since
  $Q$ is $\E$-minimal, its $\M$-subobject $W$ is also $\E$-minimal in $\K^\dual$ (\autoref{MminQuot}). The
  $\E$-minimal objects form a reflective subcategory (\autoref{MminCorefl}).
  Applying the reflection to $h\colon R\to W$, we obtain $\phi\colon V\to W$ with $h =
  \phi\cdot s$. Since $V$ is $\E$-minimal (in $\K^\dual$), $\phi$ is in $\M$,
  and since $W$ is $\M$-minimal, $\phi$ is in $\E$, and thus $\phi$ is an isomorphism.
\end{proof}

\begin{remark}
  Unfortunately, it seems very unlikely that the object obtained under both
  aspects in \autoref{propWellPointed} can be described by a universal property in $\K$.
  Given an object $C$ in $\K$, let $C'$ be the object obtained from
  \autoref{propWellPointed}. Then in general, there is neither a morphism $C\to
  C'$ nor $C'\to C$ in $\K$. This will become clear when considering an example
  coalgebra and its minimization under both aspects (\autoref{exTSwell}).
\end{remark}

In the concrete case of $F$-coalgebras, a pointed coalgebra that is both
simple and reachable is called a \emph{well-pointed} coalgebra (see
\cite[Section 3.2]{amms13}). The minimization of a pointed coalgebra under both aspects
is called the \emph{well-pointed modification}~\cite{amms13}: it is obtained by
first forming the simple quotient and then taking its reachable subcoalgebra
(i.e.~item \ref{simpleFirst} in \autoref{propWellPointed}).

\begin{instance}
  If $F\colon \C\to \C$ fulfils all assumptions from the previous
  \autoref{exExistsMinimization} (and in particular preserves inverse images),
  then the construction of the simple quotient and the reachability construction
  for $F$-coalgebras can be performed in any order, yielding the same well-pointed
  coalgebra. 
\end{instance}

In sets, the reachability computation is a simple breadth-first
search~\cite{wmkd20reachability}, and hence runs in linear time. On the other
hand, existing algorithms for computing the simple quotient for many $\Set$-functors run
in at least $n\cdot \log n$ time where $n$ is the size of the
coalgebra~\cite{GrooteEA21,WissmannDMS19}. Hence, the reachability analysis
should be done first whenever possible.

\begin{example}\label{exTSwell}
  The powerset functor $\Pow\colon\Set\to\Set$ preserves inverse images and
  arbitrary intersections, so minimization of transition systems under
  reachability and bisimilarity can be done in any order.
  \autoref{figTSWell} shows an example of a pointed transition system $C$, whose
  well-pointed modification can be obtained by performing the minimization
  aspects in any order, both yielding the one-state transition system $M$.
  Note that there is no coalgebra homomorphism between $C$ and $M$ (in neither
  direction, as indicated by $\not\to$). This indicates that the well-pointed modification of a coalgebra
  $C$ can not be described by a universal property in $\Coalg(\Pow)$.
  \begin{figure} \centering
    \begin{tikzpicture}[coalgfit/.style={
        coalgebra frame,
        inner xsep=5mm,
        inner ysep=4mm,
        xshift=-1.5mm,
        yshift=1.5mm,
      },
      ]
    \def\cellx{7cm}
    \def\celly{2.5cm}
    \begin{scope}[coalgebra,shift={(0,\celly)}]
      \node[state] (q0) {$\bullet$};
      \node[state] (q1) at (1,0) {$\bullet$};
      \node[state] (q2) at (2,0) {$\bullet$};
      \node[state] (q3) at (3,0) {$\bullet$};
      \draw[transition] ([xshift=-3mm]q0.west) to (q0.west);
      \draw[transition,bend left] (q0) to node[above] {}(q1);
      \draw[transition,bend left] (q1) to node[above] {}(q0);
      \draw[transition] (q2) to node[above] {}(q1);
      \draw[transition] (q2) to node[above] {}(q3);
      \node[coalgfit,fit={(q0) (q3)}] (C) {};
      \node[coalgname,anchor=west] at (C.north west) {$C$};
    \end{scope}
    \begin{scope}[coalgebra,shift={(\cellx,\celly)}]
      \node[state] (q0) {$\bullet$};
      \node[state] (q2) at (1,0) {$\bullet$};
      \node[state] (q3) at (2,0) {$\bullet$};
      \draw[transition] (q2) to node[above] {}(q3);
      \draw[transition] ([xshift=-3mm]q0.west) to (q0.west);
      \draw[transition,loop right,out=50,in=110,looseness=5]
      (q0) to node[right] {}(q0);
      \draw[transition] (q2) to node[above] {}(q0);
      \node[coalgfit,fit={(q0) (q2) (q3)}] (Q) {};
      \node[coalgname,anchor=west] at (Q.north west) {$Q$};
    \end{scope}
    \begin{scope}[coalgebra,shift={(0,0)}]
      \node[state] (q0) {$\bullet$};
      \node[state] (q1) at (1,0) {$\bullet$};
      \draw[transition] ([xshift=-3mm]q0.west) to (q0.west);
      \draw[transition,bend left] (q0) to node[above] {}(q1);
      \draw[transition,bend left] (q1) to node[above] {}(q0);
      \node[coalgfit,fit=(q0) (q1)] (R) {};
      \node[coalgname,anchor=west] at (R.north west) {$R$};
    \end{scope}
    \begin{scope}[coalgebra,shift={(\cellx,0)}]
      \node[state] (q0) {$\bullet$};
      \draw[transition] ([xshift=-3mm]q0.west) to (q0.west);
      \draw[transition,loop right,out=50,in=110,looseness=5]
      (q0) to node[right] {}(q0);
      \node[coalgfit,fit=(q0),label distance=5pt] (M) {};
      \node[coalgname,anchor=west] at (M.north west) {$M$};
    \end{scope}
    \begin{scope}[path/.style={
        commutative diagrams/.cd, every arrow, every label,
        shorten <= 2mm,shorten >= 2mm,
      }]
    \path[path]
    (C) edge[->>] (Q)
    (R) edge[->>] (M)
    (R) edge[>->] (C.south -| R)
    ;
    \path[path,xshift=2mm,>->] (M.north) to[>->] (M |- Q.south);
    \path[path,xshift=1mm,yshift=1mm] (C.south east) to node[anchor=center,rotate=-35,pos=0.45]{$\not\,$} (M.north west);
    \path[path,xshift=-1mm,yshift=-1mm] (M.north west) to node[anchor=center,rotate=-35,pos=0.45]{$\not\,$} (C.south east);
    \end{scope}
  \end{tikzpicture}
  \caption{Minimization of a $\Pow$-coalgebra under reachability and
    observability (i.e.~bisimilarity), with no morphisms between $C$ and $M$ ($\not\to$)}
  \label{figTSWell}
  \end{figure}
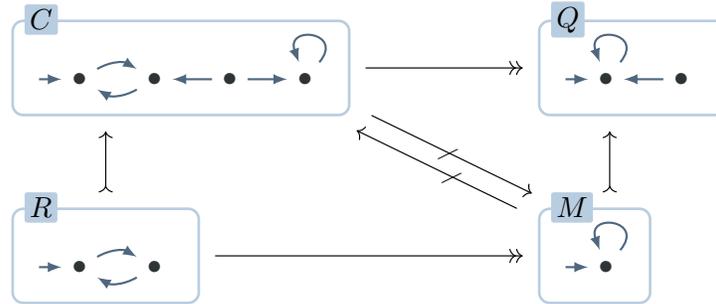

\end{example}

If $F$ does not preserve inverse images, then in the construction of the simple
quotient, transitions may cancel out each other and this may affect the
reachability of states. We have seen an example for this in
\itemref{exQuotClosed}{exQuotClosed:Reach} where performing reachability first
and observability second leads to a simple coalgebra in which states are unreachable, i.e.~the result is not well-pointed.
Hence, in
contrast to the well-known automata minimization procedure, the minimization of
a coalgebra in general has to be performed by first computing its simple quotient and secondly
computing the reachable part in the simple quotient.

\section{Conclusions}
We have seen a common ground for minimality notions in a category with various
instances in a coalgebraic setting. The abstract
results about the uniqueness and the existence of the minimization instantiate
to the standard results for reachability and observability of coalgebras.
Most of the general results even hold if the $(\E,\M)$-factorization system is
not proper. The tree unravelling of an automaton is an instance of
minimization for a non-proper factorization system.

It remains for future work to relate the efficient algorithmic approaches to the
minimization tasks: reachability is computed by breadth-first
search~\cite{wmkd20reachability,BarloccoKR19} and observability is computed by
partition refinement algorithms~\cite{KonigK14,WissmannDMS19,DorschMSW17}.
Even though their run-time complexity differs -- reachability is usually linear,
whereas partition refinement algorithms are quasilinear or slower -- they have
striking similarities.
All these algorithms compute a chain of subobjects resp.~quotients on the carrier of
the input coalgebra and terminate at the first element of the chain admitting a
coalgebra structure compatible with the input coalgebra. It is thus likely that
this relation can be made formal. A similar connection between the reachability of
algebras and partition refinement on coalgebras is already known~\cite{Rot16}.

\subsection*{Acknowledgements}%
The author thanks Stefan Milius and Jurriaan
Rot for inspiring discussions and thanks the referees for their helpful comments.
The author thanks Bálint Kocsis for finding a mistake in \autoref{figTSWell},
which is corrected in the present version.

\bibliographystyle{alphaurl}
\bibliography{refs}

\newpage
\appendix
\section{Proofs of standard results}
\begin{proofappendix}{cor:coalg-colims}
  Let $c_i\colon UDi\to FUDi$ be the coalgebra structure of $Di\in \Coalg(F)$
  for every $i\in \D$.
  Consider the colimit of $UD\colon \D\to \C$
  \[
    \begin{tikzcd}
      UDi
      \arrow{r}{\inj_i}
      & \colim (UD)
    \end{tikzcd}
    \qquad\text{for every }i\in \D
  \]
  and apply $F$ to it. Precomposition with $c_i$ yields
  \[
    \begin{tikzcd}
      UDi
      \arrow{r}{c_i}
      &
      FUDi
      \arrow{r}{F\inj_i}
      & F\colim (UD)
    \end{tikzcd}
    \qquad\text{for every }i\in \D.
  \]
  This is a cocone for the diagram $UD$ because for all $h\colon i\to j$ in $\D$
  the outside of the following diagram commutes:
  \[
    \begin{tikzcd}
      UDi
      \arrow{r}{c_i}
      \descto{dr}{$Dh$ coalgebra \\ morphism}
      \arrow{d}[swap]{UDh}
      &[8mm]
      FUDi
      \arrow{rd}{F\inj_i}
      \arrow{d}{FUDh}
      \\
      UDj
      \arrow{r}[swap]{c_j}
      &
      FUDj
      \arrow{r}[swap]{F\inj_j}
      & F\colim (UD)
    \end{tikzcd}
  \]
  Thus we obtain a coalgebra structure $u\colon \colim(UD)\to F\colim(UD)$.
  Since $u$ is a cocone-morphism, every $\inj_i$ is an $F$-coalgebra morphism.

  For any other coalgebra structure $u'\colon \colim(UD)\to F\colim(UD)$ for
  which every $\inj_i$ is an $F$-coalgebra morphism, we have $u=u'$ by the
  colimit $\colim(UD)$. Hence, $u$ is the only coalgebra structure that making
  all $\inj_i$ coalgebra morphisms.

  In order to show that $(\colim(UD), u)$ is the colimit of $D\colon \D\to
  \Coalg(F)$, consider another cocone $(m_i\colon Di\to (E,e))_{i\in \D}$.
  \[
    \begin{tikzcd}
      \colim (UD)
      \arrow{d}[swap]{u}
      \arrow[dashed]{r}{w}
      & E \arrow{d}{e} \\
      F\colim(UD)
      \arrow{r}{Fw}
      & FE
    \end{tikzcd}
  \]
  In $\C$, we obtain a cocone morphism $w\colon \colim (UD)\to E$. With a
  similar verification as before,  $(e\cdot m_i\colon UDi\to FE)_{i\in \D}$ is a
  cocone for $UD$, and thus both $e\cdot w$ and $Fw\cdot u\colon \colim (UD)\to
  FE$ are cocone morphisms (for $UD$). Since $\colim (UD)$ is the colimit, this
  implies that $e\cdot w = Fw\cdot u$, i.e.~$w\colon (\colim (UD), u)\to (E,e)$ is
  a coalgebra morphism. Since $U\colon \Coalg(F)\to \C$ is faithful, $w$ is the
  unique cocone morphism, and so $(\colim UD,u)$ is indeed the colimit of~$D$.
  \qed
\end{proofappendix}

\begin{proofappendix}{EM:pullbackwide}
    Consider the $(\E,\M)$-factorization of $\pr_j$ into $e\colon
    P\epito C$ and $m\colon C\monoto A_j$ with $\pr_j = m\cdot e$. On the image,
    we define a cone structure $(c_i\colon C\to A_i)_{i\in I}$ by $c_j = m$ and
    for every $i\in I\setminus\{j\}$ by the diagonal fill-in:
    \[
      \begin{tikzcd}
        P
        \arrow[->>]{r}{e}
        \arrow{d}[swap]{\pr_i}
        & C
        \arrow[>->]{r}{c_j}
        \arrow[dashed]{dl}{c_i}
        & A_j
        \arrow{d}{f_j}
        \\
        A_i
        \arrow[>->]{rr}{f_i}
        & & B
      \end{tikzcd}
      \quad\text{for all }i\in I\setminus\{j\}.
    \]
    The diagonal $c_i$ is induced, because $f_i\in \M$ for all $i\in I\setminus\{j\}$.
    The family $(c_i)_{i\in I}$ forms a cone for the wide pullback, because for
    all $i,i'\in I$ we have $f_i\cdot c_i = f_j\cdot c_j = f_{i'}\cdot c_{i'}$.
    This makes $e$ a cone morphism, because $c_i\cdot e = \pr_i$ for
    all $i\in I$.
    Moreover, the limiting cone $P$ induces a cone morphism $s\colon C\to P$ and
    we have $s\cdot e = \id_P$. Consider the commutative diagrams:
    \[
      \begin{tikzcd}
        P
        \arrow[->>]{rr}{e}
        \arrow{dd}[swap]{e}
        \arrow{dr}[description,shape=circle,inner sep=1pt]{\id_P}
        & & C
        \arrow{dl}[description,shape=circle,inner sep=2pt]{s}
        \arrow{dd}{c_j}
        \\
        & P
        \arrow{dr}[description,shape=circle,inner sep=2pt]{\pr_j}
        \arrow{dl}[description,shape=circle,inner sep=2pt]{e}
        \descto{r}{\((*)\)}
        \descto{d}{\((*)\)}
        & {}
        \\
        C
        \arrow[>->]{rr}[swap]{c_j}
        & {} & A_j
      \end{tikzcd}
      \quad\text{ and }\quad
      \begin{tikzcd}
        P
        \arrow[->>]{rr}{e}
        \arrow{dd}[swap]{e}
        & & C
        \arrow{dd}{c_j}
        \arrow{ddll}[description,shape=circle,inner sep=1pt]{\id_C}
        \\
        & \phantom{P}
        \\
        C
        \arrow[>->]{rr}{c_j}
        & {} & A_j.
      \end{tikzcd}
    \]
    The parts marked by $(*)$ commute because $e$ and $s$ are cone morphisms.
    Since the diagonal fill-in in \itemref{D:factSystem}{diagonalization} is unique, we have $e\cdot s = \id_C$. Thus,
    $e$ is an isomorphism, and $\pr_j = c_j\cdot e$ is in $\M$, as desired.
    \qed
\end{proofappendix}

\end{document}